%% file: main.tex
\documentclass[envcountsame]{llncs}

\usepackage{amssymb}
\usepackage{url}
\usepackage{amsmath}
\usepackage{graphicx}
\usepackage{fancyvrb}

\usepackage[boxed,ruled,vlined]{algorithm2e}

\usepackage{multirow}

\usepackage{chngpage}

\newcommand{\hide}[1]{}
\newcommand{\eg}{{\em e.g.}}
\newcommand{\ie}{{\em i.e.}}
\newcommand{\preach}{{\sffamily ProbReach}}

\begin{document}

\title{Probabilistic bounded reachability for hybrid systems with continuous nondeterministic 
and probabilistic parameters}

\author{Fedor Shmarov \and Paolo Zuliani}
\institute{School of Computing Science, Newcastle University, Newcastle upon Tyne, UK \\
\email{\{f.shmarov, paolo.zuliani\}@ncl.ac.uk}}

\maketitle

\begin{abstract}
We develop an algorithm for computing bounded reachability
probability for hybrid systems, i.e., the probability that the system reaches an unsafe
region within a finite number of discrete transitions. In particular, we focus on hybrid
systems with continuous dynamics given by solutions of nonlinear ordinary 
differential equations (with possibly nondeterministic initial conditions and parameters), and 
probabilistic behaviour given by initial parameters distributed as continuous (with 
possibly infinite support) and discrete random variables. 
Our approach is to define an appropriate relaxation of the (undecidable) reachability problem, 
so that it can be solved by $\delta$-complete decision procedures. In particular, for systems 
with continuous random parameters only, we develop a validated integration procedure which 
computes an arbitrarily small interval that is guaranteed to contain the reachability probability. 
In the more general case of systems with both nondeterministic and probabilistic parameters, our 
procedure computes a guaranteed enclosure for the range of reachability probabilities. 
We have applied our approach to a number of nonlinear hybrid models and validated the results 
by comparison with Monte Carlo simulation.


\end{abstract}

\section{Introduction}
\input{intro}

\section{Probabilistic Bounded $\delta$-Reachability}
\input{deltareach}



\input{probreachimpl}



\input{algorithm}

\input{results}

\vspace{-1ex}
\section{Conclusions and Future Work}
We have given a formal definition of the bounded probabilistic $\delta$-reachability problem
for hybrid systems with continuous random and nondeterministic initial parameters. We have combined
validated integration with $\delta$-complete decision procedures for solving the 
probabilistic $\delta$-reachability problem. Our technique computes a {\em numerically guaranteed}
enclosure for the probabilities that the system reaches the unsafe region in a finite number of 
discrete transitions. For systems with continuous random (but no nondeterministic) parameters,
such enclosure can be made arbitrarily small.
We have implemented our technique in the open source tool \preach\  and have applied
it to a number of case studies featuring highly nonlinear ODEs, unbounded continuous random 
parameters and nondeterministic parameters. We have validated our results against Monte Carlo
simulation, and the comparison supports the correctness of our approach. 
\hide{
Also, our technique can produce faster and accurate results, returning intervals 
of size $10^{-9}$ even when computing very small probabilities (see, \eg,
the results for the bouncing ball model). This is a significant advantage with respect 
to Monte Carlo verification techniques such as statistical model checking
\cite{YounesS06}. It is well-known that Monte Carlo methods suffer from the rare-event 
problem: to estimate reliably very small probabilities, 
extremely large sample sizes (\ie, system simulations) are needed \cite{RRbook,Z:hscc12}.
}

Our work shows that it is possible to verify bounded reachability for hybrid systems 
featuring continuous random and nondeterministic parameters with the same level of accuracy as
for finite-state stochastic systems. Of course, more work needs to be done in terms of improving both 
the tool engineering and the theory. With respect to the former, a more efficient parallel strategy
needs to be implemented, and more experiments need to be performed to assess better the tool scalability. 
For the theory, in the future we plan to tackle a larger class of hybrid systems, which in particular
include state-dependent probabilistic jumps and continuous probabilistic dynamics (stochastic 
differential equations).

\bibliographystyle{splncs03}
\vspace{-3ex}
\bibliography{../refs}

\newpage
\appendix
\section*{Appendix}
\input{appendix}

\end{document}

%% file: intro.tex
Cyber-physical systems integrates digital computing (the {\em cyber} part) with a {\em physical} 
environment or device, in order to enhance or enable new capabilities 
of physical systems. Hybrid systems are mathematical models that combine continuous dynamics 
and discrete control, and enjoy widespread use for modelling cyber-physical systems.
For example, Stateflow/Simulink\footnote{\url{www.mathworks.com/simulink}} is the {\em de facto} 
standard tool for model-based design of embedded systems, and its semantics can be given 
in terms of hybrid systems (\eg, \cite{Tiwari02}). Cyber-physical systems are
used in many safety-critical applications, where a malfunctioning can result in threats to, 
or even loss of, human life. For example, modern aircraft are flown more efficiently 
by a computer, while anti-lock brakes and stability control
contribute to safer cars. Again, electronic biomedical devices (\eg, digital infusion pumps) 
offer superior flexibility and accuracy than traditional devices. Thus, verifying safety of 
cyber-physical systems, and thereby of hybrid systems, is an extremely important problem.

The state space of a hybrid system consists of a discrete component and of a continuous component. 
The fundamental {\em reachability} problem is to decide whether a hybrid system reaches an 
{\em unsafe} region of its state space (a subset of states indicating incorrect behaviour of 
the system). Unfortunately, this problem is undecidable even for hybrid systems with constant 
differential dynamics \cite{DBLP:conf/hybrid/AlurCHH92}. 
For timed automata, \ie, same constant differential dynamics across all the variables, 
the reachability problem is PSPACE-complete \cite{DBLP:conf/icalp/AlurD90}. Also, it has been 
recently shown that bounded-time reachability of rectangular automata with non-negative rates 
is decidable \cite{BrihayeDGORW11}.
However, hybrid systems arising from practical applications feature much richer dynamics, including
non-linear functions over the reals, \eg, trigonometric functions,
for which even simple questions are in general undecidable \cite{Richardson68}.
Furthermore, for many practical applications it is necessary to augment hybrid systems with 
stochastic behaviour. Stochastic systems arise naturally when modelling phenomena which are 
intrinsically probabilistic, \eg, soft errors in computing hardware. Also, stochastic systems 
can arise due to uncertainty in (deterministic) system components, its behaviours, and its 
environment. The reachability problem for stochastic hybrid systems asks what is the {\em probability} 
that the system reaches the unsafe region. (Note that for hybrid systems with both stochastic and 
non-deterministic behaviour the answer may be a range of probabilities.)
In this work we focus on {\em bounded} reachability, \ie, within a finite number of discrete 
transitions.

Since even standard reachability is undecidable, the problem must be modified if we want 
to solve it algorithmically. A possible solution is to relax it in a sound manner through
the notions of $\delta$-satisfiability and $\delta$-complete decision 
procedures \cite{DBLP:conf/lics/GaoAC12}. Such procedures sidestep undecidability by allowing a
`tuneable' precision in the answer provided. This is a necessary condition for decidability, 
and it motivates the notion 
of $\delta$-satisfiability for logical formulae over the reals \cite{DBLP:conf/lics/GaoAC12}. 
Using $\delta$-satisfiability, in this paper we introduce and study the notion of probabilistic 
$\delta$-reachability.

To summarise, in this paper:
\begin{itemize}
	\vspace{-1ex}
	\item we formulate the bounded $\delta$-reachability problem for hybrid systems 
		with continuous/discrete probabilistic and nondeterministic initial parameters;
	\item we develop an algorithm that combines validated integration and $\delta$-complete 
		procedures into for computing a {\em numerically guaranteed} enclosure for the 
		reachability probabilities. For models with continuous random 
		(but no nondeterministic) parameters, such enclosure can be made {\em arbitrarily} small;
	\item we validate our algorithm against standard Monte Carlo probability estimation 
		on a number of case studies.
\end{itemize}

\paragraph{{\bf Related Work.}}
The SiSAT tool \cite{Sisat} solves probabilistic bounded 
reachability by returning answers guaranteed to be numerically accurate.
However, SiSAT does not currently support continuous random parameters, while instead our
tool does so (also with unbounded domains, \eg, normal random variables). A very recent extension 
of SiSAT supports continuous nondeterminism, but the technique is based on statistical model checking 
and therefore can only provide statistical guarantees \cite{SMCnd}, while we give numerical and formal
guarantees.
In \cite{Enszer10verifiedsolution} the authors present a technique for computing p-boxes using validated ODE
integration. However, the technique is restricted to ODE systems and finite-support random parameters, 
while we handle hybrid systems and infinite-support random parameters. Moreover, it is 
not clear what guarantees are given for models containing only continuous and/or discrete random parameters: 
the size of the computed p-box might be quite large. In contrast, for continuous random parameters we can 
compute an arbitrarily small interval containing the exact reachability probability (see Proposition \ref{prop:cn_nd}).

UPPAAL \cite{UPPAAL} is an 
extremely powerful model checker for timed automata, and it has been recently extended to support 
(dynamic) networks of stochastic timed automata via UPPAAL SMC \cite{UPPAALSMCtut}. However,
UPPAAL SMC utilises a statistical model checking approach for reasoning
about probabilities. PRISM \cite{KNP11} is a state-of-the-art model checker for a variety
of discrete-state stochastic systems, but with respect to real-time systems it is limited to
probabilistic timed automata. The tool {FAUST}$^2$ \cite{faust2} utilises abstraction
techniques to verify nondeterministic continuous-state Markov models, although currently
for discrete-time models only. ProHVer computes an upper bound for the maximal reachability 
probability \cite{ZhangSRHH10}, and handles continuous random parameters via discrete overapproximation 
only \cite{FraenzleHHWZ11}. We instead provide an enclosure ({\em both} upper and lower bounds) of the 
whole range of probabilities (for models with nondeterministic continuous parameters); in the case of 
continuous random parameters our enclosure can be arbitrarily tight (see Proposition \ref{prop:cn_nd}).
In \cite{Abate2010624} the authors introduce a technique for computing bounds on reachability probababilities 
for stochastic hybrid systems, using abstraction by discrete-time Markov chains.
The technique is further extended to full LTL and nondeterminism \cite{TkachevA13}.
In \cite{RamponiCSL10} the authors give model checking algorithms for PCTL formulae over continuous-time 
stochastic hybrid systems. However, in \cite{Abate2010624,TkachevA13,RamponiCSL10} continuous state space 
is handled through finite discretisation and approximated numerical solutions are provided for the experiments. 
We instead consider continuous time and space, and give full mathematical/numeric guarantees.

With respect to $\delta$-satisfiability, in \cite{Novak1992} the authors
introduced and studied the complexity of a {\em relaxed} version of the verification problem, \ie, verifying whether
a given candidate is close to a problem solution. The (strong) verification problem is 
undecidable in general, so the authors relax it by introducing a ``safety zone'' in
which either answer is deemed correct --- this is $\delta$-satisfiability. 

Finally, in our work we use verified integration techniques (for an overview see, \eg, 
\cite{Petras:VerifiedNI} and references therein). Integration methods in the literature work 
with integrands in explicit form, \ie, one must provide the actual mathematical expression 
for the integrand. Our approach is more general because:
a) the integrand is given as a function of the numerical solution of possibly nonlinear ODEs; 
b) it considers hybrid dynamics.  Our algorithm carries over the guarantees provided by 
$\delta$-complete procedures for aspects a) and b) to the verified computation of a 
multi-dimensional integral over a possibly unbounded domain.

%% file: deltareach.tex
The following definition of hybrid system is a slight variant of the standard one.
\begin{definition}\label{def:HS}
A hybrid system with {\em probabilistic} and {\em nondeterministic} initial parameters 
consists of the following components:
\begin{itemize}
\item $Q = \{q_0, ..., q_m\}$ a set of modes (discrete components of the system),
\item $D = D_0 \times \ldots \times D_p $ a domain of {\em discrete} random parameters, where 
	each $D_i$ is a finite set of reals,
\item $R = [r_1, s_1] \times \ldots \times [r_l, s_l] \subset \mathbb{R}^{l}$ a domain of 
	{\em continuous} random parameters,
\item $Z = [y_1, z_1] \times \ldots \times [y_o, z_o] \subset \mathbb{R}^{o}$ a domain of 
	{\em nondeterministic} parameters,
\item $X = [u_1, v_1] \times \ldots \times [u_n, v_n] \times [0, T] \subset \mathbb{R}^{n+1}$ a domain 
	of continuous variables,

\item $S = Q \times X$ is the {\em hybrid state space} of the system,
\item $\Lambda = D \times R \times Z$ is the {\em parameter space} of the system,
\item $U \subseteq S$ an {\em unsafe} region of the state space,
\end{itemize}
and predicates (or relations)
\begin{itemize}
	\item $\text{flow}_{q}(\lambda, \textbf{x}^{0}, \textbf{x}^{t})$ 
	mapping the parameter $\lambda \in \Lambda$ and the continuous state 
	$\textbf{x}^{0}$ at time 0 to state $\textbf{x}^{t}$ at time point $t \in [0, T]$ in mode $q$
	\item $\text{init}_q(\textbf{x}^{0})$ indicating that $s = (q, \textbf{x}^{0})$ belongs 
	to the set of initial states,
	\item $\text{jump}_{q \rightarrow q'}(\lambda, \textbf{x}^{t}, \textbf{x}^{0})$ indicating 
	that the system with parameter $\lambda\mathord{\in}\Lambda$ can make a transition from 
	mode $q$, upon reaching 
	the jump condition in continuous state $\textbf{x}^{t}$ at time $t \in [0, T]$, 
	to mode $q^\prime$ and setting the continuous state to $\textbf{x}^{0}$,
	\item $\text{unsafe}_{q}(\textbf{x}^{t})$ indicating that $s = (q, \textbf{x}^{t}) \in U$
\end{itemize}
For all $q\in Q$ the sets defined by $\text{flow}_q, \text{init}_q, \text{jump}_q$,
and $\text{unsafe}_q$ are Borel; $\text{flow}_q$ and $\text{jump}_q$ are restricted 
to be functions of $(\lambda,\textbf{x}^0)$ and $(\lambda,\textbf{x}^t)$, respectively.
\end{definition}
The parameters in $\Lambda$ are assigned in the initial mode and remain 
unchanged throughout the system's evolution. Also, the
Borel assumption for the sets defined by the predicates is a theoretical requirement for well-definedness
of probabilities, and in practice it is easily satisfied. 
The continuous dynamics of the system is defined in each flow, and it can either be presented 
as a system of Lipschitz-continuous ODEs or explicitly. In this paper we focus on hybrid 
systems for which in each mode only one jump is allowed to take place (of course the model may have 
multiple jumps, but only one jump should be enabled at any time).
Given an initial value of the parameters, the semantics of a hybrid system can be informally 
thought as piece-wise continuous.
(More details about the formal semantics can be found in \cite{DBLP:conf/hybrid/AlurCHH92}.)

Bounded reachability asks whether the system reaches the unsafe region after $k \in \mathbb{N}$ 
discrete transitions.
\begin{definition}\cite{DBLP:journals/corr/GaoKCC14}\label{def:probreach}
The bounded $k$-step reachability property for hybrid systems with initial parameters
is the bounded $\Sigma_1$ sentence $ \exists \lambda\in\Lambda\  \psi (\lambda)$,
where
\begin{equation} \label{eq:reachability}
\begin{split}
\psi (\lambda) = \exists \textbf{x}^{0}_{0,q_0}, \exists \textbf{x}^t_{0,q_0}, ... , \exists \textbf{x}^{0}_{0,q_m}, \exists \textbf{x}^t_{0,q_m} , ... , \exists \textbf{x}^{0}_{k,q_m}, \exists \textbf{x}^t_{k,q_m}:\\
(\bigvee_{q \in Q} (\text{init}_{q}(\textbf{x}^{0}_{0,q}) \wedge \text{flow}_{q}(\lambda, \textbf{x}^{0}_{0,q},\textbf{x}^t_{0,q}))) \\
\wedge (\bigwedge_{i=0}^{k-1} (\bigvee_{q,q' \in Q}(\text{jump}_{q \rightarrow q'}(\lambda, \textbf{x}^t_{i,q}, \textbf{x}^{0}_{i+1,q'}) \\
\wedge (\text{flow}_{q'}(\lambda, \textbf{x}^{0}_{i+1,q'}, \textbf{x}^t_{i+1,q'}))) \wedge (\bigvee_{q \in Q} \text{unsafe}_{q}(\textbf{x}^t_{k,q}))))
\end{split}
\end{equation}
\end{definition}
Informally, the formula $\exists \lambda \in \Lambda\;\;  \psi(\lambda)$ encodes the sentence 
``there exists a parameter vector for which starting from {\em init} and following {\em flow} and {\em jump},
the system reaches the unsafe region in $k$ steps''.
We obtain reachability {\em within} $k$ steps by forming a disjunction of formula (\ref{eq:reachability})
for all values from 1 to $k$.
The bounded reachability problem can be solved using a $\delta$-complete decision 
procedure \cite{DBLP:conf/lics/GaoAC12}, which will {\em correctly} return one of the following answers: 
\begin{itemize}
	\item \textbf{unsat}: meaning that formula (\ref{eq:reachability}) is unsatisfiable 
	(the system never reaches the bad region $U$);
	\item \textbf{$\delta$-sat}: meaning that formula (\ref{eq:reachability}) is 
	$\delta$-satisfiable. In this case a witness, \ie, an assignment for all 
	the variables, is also returned.
\end{itemize}
With a $\delta$-complete decision procedure, an \textbf{unsat} answer can always be trusted, 
while a \textbf{$\delta$-sat} answer might in fact be a false alarm caused by the overapproximation.
(In Appendix \ref{apndx:delta-sat} we provide a short overview of $\delta$-satisfiability.)

We now associate a probability measure to the random parameters, and we consider the following problem: 
what is the probability that a hybrid system with initial parameters reaches the unsafe region in $k$ steps?
Note that hybrid systems with both random and nondeterministic parameters will 
feature a range of reachability probabilities (although not necessarily a full interval). 
\begin{definition}\label{def:probboundreach}
The {\em probabilistic bounded $k$-step reachability problem} for hybrid system with initial parameters is
to compute an interval $[a,b]$ such that:
\begin{equation}\label{eq:probabilistic-reachability}
	\forall z_0\in Z \quad \int_{B|_{z_0}} \,dP \  \in \  [a,b]
\end{equation}
where 
\begin{equation} \label{def:Borel-set}
B = \{\lambda\in \Lambda : \psi(\lambda)\}
\end{equation}
and formula $\psi(\lambda)$ is per Definition \ref{def:probreach}; $P$ is the probability measure 
associated with the random parameters; and $B|_{z_0}$ is the restriction of $B$ to $z_0$.
\end{definition}
Informally, $B$ is the set of the parameter values for which the system reaches the unsafe region 
in $k$ steps. 
\begin{proposition}\label{prop:BBorel}
The set $B$ defined by (\ref{def:Borel-set}) is Borel.
\end{proposition}
(Proofs can be found in Appendix \ref{apndx:proofs}.)
The proposition entails that for any choice of the nondeterministic parameters, the probability that 
the system reaches the unsafe region is well-defined, and thereby Definition \ref{def:probboundreach} is
well-posed. When consider hybrid systems with continuous random parameters only, 
Definition \ref{def:probboundreach} can be strengthened.
\begin{definition}\label{def:probboundreach_rv}
Given any $\epsilon \in \mathbb{Q} \cap (0,1]$, the {\em probabilistic bounded $k$-step reachability problem}
for hybrid systems with {\em random continuous} initial parameters and single initial state is to compute 
an interval $[a,b]$ of length up to $\epsilon$ such that:
\begin{equation}\label{eq:probabilistic-reachability_rv}
	\int_{B} \,dP \  \in \  [a,b]
\end{equation}
where 
\begin{equation} 
B = \{\lambda\in \Lambda : \psi(\lambda)\}
\end{equation}
and formula $\psi(\lambda)$ is per Definition \ref{def:probreach}; $P$ is the probability measure 
associated with the random parameters.
\end{definition}
Note that if only discrete random parameters are present it might not be possible to obtain 
an arbitrarily small enclosure.
Also, in Definition \ref{def:HS} we require all continuous domains to be bounded: 
this is a necessary condition for $\delta$-decidability of bounded reachability \cite{DBLP:conf/lics/GaoAC12}. 
However, we later show that it is still possible to reason about random parameters with 
unbounded domains, \eg, normally distributed. The key is that any probability density function
can be approximated arbitrarily well by a truncation on a large (but finite) interval.


%% file: probreachimpl.tex
\section{Validated Integration Procedure}\label{sec:vip}
We now present the first part of our $\delta$-complete procedure for calculating the $k$-step
reachability probability (\ref{eq:probabilistic-reachability_rv}).
The algorithm consists of a validated integration procedure and a decision procedure used 
for computing the set $B$ of Definition \ref{def:probboundreach_rv}. For clarity, we focus on one 
random continuous initial parameter.

\noindent{\em Notation.}
For an interval $[r]=[\underline{r}, \overline{r}] \subset \mathbb{R}$ we denote the size of
the interval by $width([r]) = \overline{r} - \underline{r}$ and by                         
$mid([r]) = \frac{\overline{r} + \underline{r}}{2}$ the central point of the interval.

Our validated integration procedure employs the (1/3) Simpson rule:
\begin{equation}\label{eq:Simpson}
\begin{split}
K([I])= \int_{a}^{b} \, f(x)\, dx\, = \frac{width([I])}{6}(f(\underline{I}) + 4 f(mid([I])) + \\
f(\overline{I})) - \frac{width([I])^5}{2880}f^{(4)}(\xi)
\end{split}
\end{equation}
where $[I]=[a,b]$, $\xi \in [I]$ and $f^{(4)}$ is the fourth derivative of an integrable 
function $f$. For our applications the integrands are probability density functions, which
satisfy the required integrability and differentiability conditions.
Our aim is to compute an interval of arbitrary size $\epsilon \in (0, 1]\cap \mathbb{Q}$ 
that contains $K$.
 
\begin{definition}
An interval extension of function $f:X \rightarrow Y$ is an operator $[\cdot]$ such that:
\begin{equation*}
	\forall x \in [r] \subseteq X: f(x) \in [f]([r]) \subseteq Y
\end{equation*}
\end{definition} 
By applying interval arithmetics, one computes interval extensions of $f$ and $f^{(4)}$.
(Interval extensions can be computed using interval arithmetics libraries, \eg, FILIB++ \cite{filib}.)
The interval version of Simpson's rule can be obtained simply by replacing in (\ref{eq:Simpson})
the occurrences of $f$ and $f^{(4)}$ with their interval extension $[f]$, and by replacing $\xi$ 
with the entire interval $I$ \cite{ValidatedIntegration}:
\[
\begin{split}
K \in [K]([I]) = \frac{width([I])}{6}([f](\underline{I}) + 4 [f](mid([I])) + 
[f](\overline{I})) - \frac{width([I])^5}{2880}[f]^{(4)}([I]).
\end{split}
\]
Furthermore, by the definition of integral:
\begin{equation}\label{eq:K}
K \in \Sigma_{i = 1}^{n} [K]([x]_{i})
\end{equation}
where the collection of $[x]_{i}$'s is a partition of $[a, b]$. Note that we require a partition in a 
measure-theoretic sense, \ie, intersections have (Lebesgue) measure 0, since these have no effect on integration.

In order to guarantee $\epsilon$-completeness of the integration it is sufficient to partition 
$[a, b]$ into $n$ intervals $[x]_{i}$ such that for each $[x]_{i}$ we have
$width([I]([x]_{i})) < \epsilon \frac{width([x]_{i})}{b - a}$. Then, the exact value $K$ of the 
integral will belong to an interval (\ref{eq:K}) of width smaller than $\epsilon$. Pseudo-code 
for the procedure computing integral (\ref{eq:Simpson}) up to an arbitrary 
$\epsilon \in (0,1]\cap\mathbb{Q}$ is given in Algorithm \ref{alg:integration}. For our
purposes we will only make use of the interval partition $T$, which will enable us 
to compute the reachability probability, \ie, integral (\ref{eq:probabilistic-reachability_rv}),
with precision $\epsilon$.
\begin{proposition}\label{prop:IntegrationComplexity}
If $f \in P_{C^5[a,b]}$, then the complexity of Algorithm \ref{alg:integration} is {\em NP}.
\end{proposition}
\begin{algorithm}
\label{alg:integration}
input: function $f$, interval $[a, b]$, $\epsilon \in \mathbb{Q}^+$\;
output: $[I]$, partition $T$ of $[a, b]$ such that $\int_a^b f \in [I]$ 
	and $width([I]) \leq \epsilon$\;
$[I] = [0.0, 0.0]$\;
$T,B = \emptyset$\;
\tcp{put initial partition on a stack}
$B.push(\{[a, b], [K]([a,b])\})$\;	
\While{$size(B) > 0$}{
	$\{[x], [y]\} = B.pop()$\;
	\If{$width([y]) > \epsilon \frac{width([x])}{b - a}$}{
		\tcp{split the interval in two}
		$B.push(\{[\underline{x}, mid([x])], [K]([\underline{x}, mid([x])])\})$\;
		$B.push(\{[mid([x]), \overline{x}], [K]([mid([x]), \overline{x}])\})$\;
	}\Else{
		\tcp{add sub-integral to the partial sum; save interval}
		$[I] = [I] + [K]([x])$\;
		$T.push([x])$\;
	}
}
\Return $T$, $[I]$\;
\caption{Validated Integration Procedure}
\end{algorithm}

\hide{
We recall that $P_{C^n[a,b]}$ denotes the class of polynomial-time (Type 2) computable functions
whose derivative $f^{(n)}$ exists and is continuous over $[a,b]$. Ko \cite[Section 6.2]{kobook}
showed that if $f$ is also analytic, then integration becomes P. However, such an algorithm
essentially uses truncated Taylor series over an arbitrary partition. Instead, 
Algorithm \ref{alg:integration} adaptively searches for a partition that guarantees the
required error bound $\epsilon$, while having a minimal number of intervals. In practice, this 
significantly benefits the performance of our whole implementation. Another advantage of 
Algorithm \ref{alg:integration} is that it does not require $f$ to be analytic.
}

%% file: algorithm.tex
\section{Computing $\delta$-Reachability Probability}

\subsection{Computing indicator functions}
From Algorithm \ref{alg:integration} we obtain a partition of the domain of the random
parameters which will guarantee the computation of integral (\ref{eq:probabilistic-reachability_rv})
with the desidered accuracy. In general, given $z_0\in Z$, the reachability probability is computed 
by integrating the probability measure of the random parameters over the restriction $B|_{z_0}$. 
We need to compute the following integral
\[
\int_{B|_{z_0}} \,dP(r)\quad  \left(= \int_{D\times R} 
        I_{U}(r, z_0)dP(r)\right)
\]
where $B$ is the set (\ref{def:Borel-set}), $z_0\in Z$, and $I_{U}$ is the indicator function
\begin{equation*}
        I_{U}(r, z_0) = 
        \begin{cases}
                1 \quad \text{if the system with parameter $(r, z_0)$ reaches
                        $U$ in $k$ steps} \\
                0  \quad\text{otherwise.}
        \end{cases}
\end{equation*}
We now show how to compute $I_U$ or, equivalently, set $B$.
Let $[\rho]\subseteq \Lambda$ be a box and $\phi$ be a formula of the form:
\begin{equation} \label{eq:phi}
\phi([\rho]) = \exists \lambda \in [\rho]\ \psi(\lambda)
\end{equation}
If the formula is true then $[\rho]$ contains a value for the initial parameters for which
the system reaches the unsafe region $U$.
Taking the complement of the unsafe region $U^{C} = S / U$ ($S$ is the state space of the 
system) and defining a predicate $\text{unsafe}^{C}_{q}(\textbf{x}^{t}) \equiv 
((q, \textbf{x}^{t}_{q}) \in U^{C})$ we want to ensure that the system never reaches the 
unsafe region {\em within} the $k$-th step with an initial parameter from $[\rho]$. 
In order to conclude that it is sufficient to evaluate the formula:
\begin{equation} \label{eq:phi_C}
\begin{split}
\phi^{C}([\rho]) = 
	\exists \lambda \in [\rho], \exists \textbf{x}^{0}_{0,q_0}, \exists \textbf{x}^{t}_{0,q_0}, \exists t_{0, q_0} , ... , \exists \textbf{x}^{0}_{0,q_m}, \\
\exists \textbf{x}^{t}_{0,q_m}, \exists t_{0, q_m}, ... , \exists \textbf{x}^{0}_{k,q_m}, \exists \textbf{x}^t_{k,q_m}, \exists t_{k, q_m}, \forall t'_{k, q_m} \in [0, t_{k, q_m}]:\\
(\bigvee_{q \in Q} (\text{init}_{q}(\textbf{x}^{0}_{0,q}) \wedge \text{flow}_{q}(\lambda, \textbf{x}^{0}_{0,q}, \textbf{x}^{t}_{0,q}))) \wedge \\
(\bigwedge_{i=0}^{k-1} (\bigvee_{q,q' \in Q}(\text{jump}_{q \rightarrow q'}(\lambda, \textbf{x}^{t}_{i,q}, \textbf{x}^{0}_{i+1,q'}) 
\wedge \\
(\text{flow}_{q'}(\lambda, \textbf{x}^{0}_{i+1,q'}, \textbf{x}^{t}_{i+1,q'}))) \wedge \\
(\bigvee_{q \in Q} (\text{unsafe}^{C}_{q}(\textbf{x}^{t}_{k,q}) \wedge (\text{jump}_{q \rightarrow q'}(\lambda, \textbf{x}^{t}_{k,q}, \textbf{x}^{0}_{k+1,q'}) \vee 
(t_{k, q_m} \ge T))))))
\end{split}
\end{equation}
Note that $\phi^C$ is not the logical negation of $\phi$ --- it is in fact an $\exists\forall$-quantified
formula.
The last term of $\phi^C$ ensures that the system either does not reach the unsafe region on 
the $k$-th step before it can make a transition to the successor mode or it reaches the time bound before reaching the unsafe region. This should not be confused with reaching the time bound in any of the preceding modes as it means that the system fails to reach the $k$-th step and should be, therefore, unsatisfiable.
If the formula evaluates to true then the system does not reach the unsafe region on the $k$-th step.
Then, set $B$ can be defined as a finite collection
$\{[\rho]_{i} : \phi([\rho]_{i}) \wedge (\neg \phi^{C}([\rho]_{i})\}$.
\hide{
By integrating a probability density function of a random variable we obtain intervals 
such that size of the enclosure containing the exact value of a partial sum for each 
of them is smaller than a local error.
Then a decision about the relation between each of the obtained intervals and the (unknown) 
Borel set $B$ should be taken. In order to achieve this we use a $\delta$-complete decision 
procedure, \ie, dReal \cite{DBLP:conf/cade/GaoKC13}, to verify formulas (\ref{eq:phi}) and 
(\ref{eq:phi_C}). 
}
To build such a collection, we iteratively evaluate $\phi$ and $\phi^C$ with a $\delta$-complete
procedure (\eg, dReal \cite{DBLP:conf/cade/GaoKC13}). Given a box $[\rho]$, there are four 
possible outcomes:
\begin{itemize}
	\item{$\phi([\rho])$ is \textbf{unsat}. Hence, there are {\em for sure} no values in $[\rho]$ 
	such that the system reaches the unsafe region, so $[\rho]$ is not in $B$.}
	\item{$\phi([\rho])$ is \textbf{$\delta$-sat}. Then, there is a value in $[\rho]$ such that 
	the system reaches $U$ or $U^{\delta}$ ($\delta$-weakening of set $U$).}
	\item{$\phi^C([\rho])$ is \textbf{unsat}. Therefore, there is {\em for sure} no value in 
	$[\rho]$ such that for all time points on the $k$-th step the system stays in $U^C$. 
	In other words, for all the values in $[\rho]$ the system 
	reaches $U$, so $[\rho]$ is fully contained in $B$.}
	\item{$\phi^C([\rho])$ is \textbf{$\delta$-sat}. Then there is a value in $[\rho]$ such that the system stays within $U^{C}$ or $U^{C^{\delta}}$. In combination with outcome \textbf{$\delta$-sat} for $\phi([\rho])$ it signals that $[\rho]$ is a {\em mixed} interval (it contains values from both 
$B$ and $B^C$).}
\end{itemize}
Therefore, \textbf{unsat} answers enable us to decide whether $[\rho]$ is a subset of or disjoint from
set $B$. If \textbf{$\delta$-sat} is returned for both formulae, then we are either dealing with a 
false alarm (an unsatisfiable formula is verified as \textbf{$\delta$-sat} because of the
overapproximation) or a mixed interval.

\vspace{-2ex}
\subsection{Main algorithm}
\vspace{-1ex}
The overapproximation (controlled by $\delta$) introduced by $\delta$-complete procedures can 
cause false alarms. We thus begin by addressing the choice of $\delta$.
Obviously, it is impossible to decide {\em correctly} (\ie, obtaining \textbf{unsat} for one of 
$\phi$ and $\phi^C$) on each interval if a fixed $\delta$ (even a very small one) is used for
evaluate all formulae. 


\begin{lemma} \label{lemma:converge}
Let $\phi$ be an arbitrary bounded $\Sigma_1$ formula and $\phi^{\delta}$ its weakening. Then the following holds:
\begin{equation*}
\forall \delta, \delta' \in \mathbb{Q}^{+}, 0 \leq \delta' < \delta: \neg \phi^{\delta} \rightarrow \neg \phi^{\delta'} \rightarrow  \neg \phi
\end{equation*}
\end{lemma}
(See Appendix \ref{apndx:delta-sat} for an overview of $\delta$-weakening.)
Lemma \ref{lemma:converge} means that unsatisfiability of a weakened formula 
implies unsatisfiability of its strengthening and of the initial formula. 
We next show that when an interval is uncertain, by applying Lemma \ref{lemma:converge} 
we can obtain $\delta$ and a subinterval for which 
a $\delta$-complete decision procedure can give a correct answer.
\begin{proposition} \label{prop:converge}
Let $\phi$ and $\phi^C$ as per (\ref{eq:phi}) and (\ref{eq:phi_C}), and $[u,v]$ an interval. Then:
\begin{equation*}
\begin{split}
&\exists \delta \in \mathbb{Q}^{+}:(\phi([u, v])-\delta\textbf{-sat}) \wedge (\phi^{C}([u, v])-\delta\textbf{-sat}) \Rightarrow \\
&\exists [u', v'] \subseteq [u, v]: (\phi([u', v'])-\textbf{unsat}) \oplus (\phi^{C}([u', v'])-\textbf{unsat})
\end{split}
\end{equation*}
where $\oplus$ denotes exclusive or.
\end{proposition}


We now present the full algorithm for computing bounded reachability probability. We begin by addressing
random initial parameters with (un)bounded support. Given $\epsilon \in (0,1]\cap\mathbb{Q}$, it is always 
to possible to find a {\em bounded} region of the random variable support with area larger
than $1-\epsilon$. In fact, such a problem can be stated as a $\delta$-satisfiability question 
and thus solved by a $\delta$-complete procedure. Therefore, the verified integration procedure presented
in Section \ref{sec:vip} can be applied to a random variable with unbounded domain.
If we introduce multiple independent random parameters we can still use the same verified 
integration procedure provided that each random variable is integrated with a higher accuracy, as the
next proposition shows.
\begin{proposition}\label{prop:multrv}
Given a hybrid system with $l$ independent continuous random parameters, to compute with precision 
$\epsilon_{prod} \in (0,1]\cap\mathbb{Q}$ the reachability probability it is sufficient that each random 
variable is integrated with precision $\epsilon$ satisfying:
\begin{equation} \label{eq:multivarmaininequalityspecialcase}
\epsilon_{prod} \ge \sum_{i=1}^l {l\choose{i}}\epsilon^i 
\end{equation}
where ${l}\choose{i}$ is the binomial coefficient.
\end{proposition}

Suppose now a hybrid system has (continuous) nondeterministic parameters. Then the 
probability that the system reaches the unsafe region becomes a function of the nondeterministic 
parameters. In particular, the indicator function 
$I_{U}(r, z)$ can be equal to 0 and 1 for the same values of the continuous
{\em random} parameters, \ie, there may exist $r_0$ and $z_0 \neq z_1$ such that 
$I_{U}(r_0, z_0) = 0$ and $I_{U}(r_0, z_1) = 1$. 
Therefore, it is in general impossible to provide any guarantees on the length of probability interval,
and we need to compute an enclosure for all probabilities.
We will use the following symbolic notation for hybrid systems:
\begin{itemize}
	\item{\textbf{HA} ({\em Hybrid Automaton})} - a hybrid system without initial random parameters 
		(only deterministic and nondeterministic).
	\item{\textbf{PHA} ({\em Probabilistic Hybrid Automaton})} - a hybrid system with random and 
		deterministic continuous initial parameters (no nondeterminism).
	\item{\textbf{NPHA} ({\em Nondeterministic Probabilistic Hybrid Automaton})} - a hybrid system 
		with random, deterministic and nondeterministic continuous initial parameters.
\end{itemize}
We first state the algorithm for \textbf{NPHA}s with no discrete probability.
\begin{proposition}\label{prop:cn_nd}
Given $\epsilon \in (0, 1] \cap \mathbb{Q}$, $k\in\mathbb{N}$, and an {\em\textbf{NPHA}} {\em without discrete
random parameters}, there exists an algorithm for computing an interval containing the set of $k$-step 
reachability probabilities. If the system has no nondeterministic parameters, the algorithm returns an 
interval of size not larger than $\epsilon$ containing the $k$-step reachability probability 
(\ref{eq:probabilistic-reachability_rv}).
\end{proposition}
The pseudo-code of the algorithm is presented in Algorithm \ref{alg:prob-reach}. Informally,
the algorithm starts by getting an interval partition from the validated integration
procedure (Algorithm \ref{alg:integration}) for each random variable; also, a candidate
probability interval is initialised to $[0,1]$. 
Then, it evaluates the formulae $\phi$ and $\phi^C$ on the current partition, which will be 
refined whenever both $\phi$ and $\phi^C$ are $\delta$-\textbf{sat}. Instead, an \textbf{unsat} 
answer is used to refine the probability interval. The termination condition depends on the model type.
If there are no nondeterministic parameters, then the algorithm will terminate when the width 
of the probability interval satisfies the desired size $\epsilon$. Otherwise, the algorithm 
terminates when the maximum length of the boxes in the partition is smaller than $\epsilon$. 
(Given a box we can split it into $2^{n}$ boxes of (pairwise) equal size in such a way that each 
interval in the box is reduced. However, any division strategy can be 
applied as long as the size of each interval forming the box is reduced.)

\begin{algorithm}
\label{alg:prob-reach}
input: continuous random parameters $\overline{r} = \{r_{1}, ..., r_l\}$ with their probability 
densities $\overline{f}(r)$, $\epsilon_{prod}\in (0,1]\cap\mathbb{Q}$, hybrid system description $\phi$\;
output: interval $[I]$ enclosing the reachability probabilities\\
\tcp{obtain $\epsilon$ from (\ref{eq:multivarmaininequalityspecialcase}) using $\epsilon_{prod}$}
$\epsilon_{inf} = t \epsilon$\;
$\epsilon_{prob} = (1-t) \epsilon$\;

\tcp{obtain bounds from (\ref{eq:epsilon-inf}) for each continuous random parameter}
$[\overline{a}, \overline{b}] = \bigcup_{i=1}^{l}($\textbf{bounds}$(f(r_{i}), \epsilon_{inf})$)\;

$[\overline{r}] = \bigcup_{i=1}^{l}($\textbf{Algorithm1}$(f(r_{i}), [a_{i}, b_{i}], \epsilon_{prob})$)\;

\tcp{get Cartesian product of intervals from obtained partitions}
$B$.push$([r_{1}]\times \cdots \times[r_{l}])$\;
\tcp{set initial probability intervals for upper/lower approximation}
$[P_{lower}] = [0.0, 0.0]$\;
$[P_{upper}] = [1.0, 1.0]$\;
\tcp{consider unbounded segments}
$[P_{upper}] = [P_{upper}] + 1 - \prod_{i=1}^{l}\int_{a_{i}}^{b_{i}}\, f_{i}(x)\, dx $\;
\While{true}
{
	\tcp{stack containing extra divisions of the boxes}
	$D = \emptyset$\;
	\While{$size(B) > 0$}
	{
	$\textbf{box} = B$.pop()\;
	\tcp{$\delta$-complete procedure evaluates formula}
	\If{$\phi(\textbf{box})$ = $\delta$-sat}
	{
		\tcp{$\delta$-complete procedure evaluates formula}
		\If{$\phi^{C}(\textbf{box})$ = $\delta$-sat}
		{
			\tcp{split initial box into $2^{l}$ boxes of equal size}
			$D$.push(branch$(\textbf{box})$)\;
		}\Else
		{
			\tcp{increase lower approximation}
			$[P_{lower}] = [P_{lower}] + [S](\textbf{box})$\;
		}
	} \Else
	{
		\tcp{decrease upper approximation}
		$[P_{upper}] = [P_{upper}] - [S](\textbf{box})$\;
	}
	}
	$B = D$\;
	\If{MODEL\_TYPE($\phi$) = \textbf{PHA}}
	{
		\tcp{termination condition when nondeterminism is absent}
		\If{$\overline{[P_{upper}]} - \underline{[P_{lower}]} \le \epsilon_{prob}$}
		{
			\Return $[\underline{[P_{lower}]}, \overline{[P_{upper}]}]$\;
		}
	}
	\If{MODEL\_TYPE($\phi$) = \textbf{NPHA}}
	{
		\If{$\max_{\textbf{box} \in B}(|\textbf{box}|) \le \epsilon$}
		{
			\Return $[\underline{[P_{lower}]}, \overline{[P_{upper}]}]$\;
		}
	}	
}

\caption{Probabilistic $\delta$-reachability \textbf{PHA} and \textbf{NPHA}}
\end{algorithm}

\begin{theorem}\label{thm:dd}
Given $\epsilon \in (0, 1] \cap \mathbb{Q}$, $k\in\mathbb{N}$, and a full {\em\textbf{NPHA}},
there exists an algorithm for computing an interval containing the set of $k$-step reachability 
probabilities. If the system has no nondeterministic parameters, 
the algorithm returns an interval of size not larger than $\epsilon$ containing the
$k$-step reachability probability (\ref{eq:probabilistic-reachability_rv}).
\end{theorem}
Algorithm \ref{alg:gen-prob-reach} drives the whole verification loop, while also handling
discrete random parameters (with essentially the same technique as before). Notice that when
the model has continuous parameters, Algorithm \ref{alg:prob-reach} is utilised.

\begin{algorithm}
\label{alg:gen-prob-reach}
input: continuous random parameters $\overline{r} = \{r_1, \ldots , r_l\}$ with probability 
densities $\overline{f}(r)$, discrete random parameters $\{D_{1}, \ldots, D_p\}$ with probability 
distributions $\overline{p}(\cdot)$, 
$\epsilon_{prod}\in (0,1]\cap\mathbb{Q}$, hybrid system description $\phi$\;
output: interval $[I]$ enclosing reachability probability

\tcp{obtain Cartesian product of discrete random parameters}
$DD$.push$(D_{1}\times \cdots \times D_p)$\;
\tcp{set initial probability intervals for upper/lower approximation}
$P = \emptyset$\;
\While{$size(DD) > 0$}
{
	$\textbf{dd} = DD$.pop()\;
	$m_{\textbf{dd}} = \prod_{d \in \textbf{dd}}p(d)$\;
	\tcp{syntactically replace discrete parameters and get new model}
	\tcp{if no other parameters are present, $\phi_\mathbf{dd}$ becomes an \textbf{HA}}
	$\phi_\textbf{dd} = \phi [(D_1, \ldots D_p) \backslash \textbf{dd}$]\;
	\If{MODEL\_TYPE($\phi_\mathbf{dd}$) = \textbf{HA}}
	{
		\tcp{$\delta$-complete procedure evaluates formula}
		\If{$\phi_\mathbf{dd}$ = $\delta$-sat}
		{
			\tcp{$\delta$-complete procedure evaluates formula}
			\If{$\phi^{C}_\mathbf{dd}$ = $\delta$-sat}
			{
				\tcp{could not correctly decide for $\phi_\mathbf{dd}$}
				$P$.push($m_{\textbf{dd}} \cdot [0.0, 1.0]$)\;
			}\Else
			{
				\tcp{formula $\phi_\mathbf{dd}$ is sat}
				$P$.push($m_{\textbf{dd}} \cdot [1.0, 1.0]$)\;
			}
		} \Else
		{
			\tcp{formula $\phi_\mathbf{dd}$ is unsat}
			$P$.push($m_{\textbf{dd}} \cdot [0.0, 0.0]$)\;
		}
	}
	\If{(MODEL\_TYPE($\phi_\mathbf{dd}$) = \textbf{PHA} $\vee$ \textbf{NPHA})}
	{
		$[I] = \textbf{Algorithm 2}(\overline{r}, \epsilon_{prod}, \phi_{\textbf{dd}})$\;
		\tcp{add obtained probability interval to the stack}
		$P$.push($m_{\textbf{dd}} \cdot [I]$)\;
	}
}
\tcp{obtain sum of all probability intervals on the stack}
$[\underline{[P_{lower}]}, \overline{[P_{upper}]}] = \sum_{[I] \in P}[I]$\;
\Return $[\underline{[P_{lower}]}, \overline{[P_{upper}]}]$\;
\caption{Main \preach\ algorithm}
\end{algorithm}

\begin{theorem}\label{thm:complexity}
The complexity of Algorithm \ref{alg:prob-reach} is ${NP}^{(\Sigma_{2}^{P})^{C}}$, where 
$P \subseteq C \subseteq PSPACE$ is the complexity of the terms in the description of the
hybrid system. With Lipschitz-continuous ODEs terms the complexity is $PSPACE$-complete.
\end{theorem}

%% file: results.tex
\vspace{-3ex}
\section{Experiments}\label{sec:Exp}
\vspace{-1ex}

We have implemented our algorithms in \preach; its source code 
and the models studied are on \texttt{https://github.com/dreal/probreach}. 
(The tool implementation is explained in \cite{ProbReach}.) 
The results below can be also accessed on 
\texttt{https://homepages.ncl.ac.uk/f.shmarov/probreach}. All experiments were
carried out on a multi-core Intel Xeon E5-2690 2.90GHz system running Linux Ubuntu 14.04LTS. 
The algorithms were also parallelised, and the results below feature the \preach\  CPU time of the
parallel version on 24 cores.

We have applied \preach\ to four hybrid models: a 2D-moving bouncing ball, human starvation, prostate
cancer therapy, and car collision scenario. The models feature a variety of highly nontrivial dynamics.
For example, the ODEs for the prostate cancer therapy model \cite{Liu2015} include exponential terms:
\[
\begin{split}
\frac{dx}{dt} = &\left(\frac{\alpha_{x}}{1+\mathrm{e}^{(k_1-z)k_2}}-\frac{\beta_{x}}{1+\mathrm{e}^{(z-k_3)k_4}}-m_1\left(1-\frac{z}{z_0}\right)-c_1\right)x+c_2 \\
\frac{dy}{dt} = &m_1\left(1-\frac{z}{z_0}\right)x + \left(\alpha_{y}\left(1-d_0\frac{z}{z_0}\right)-\beta_{y}\right)y \\
\frac{dz}{dt} = &-z \gamma - c_3
\end{split}
\]
More details on the models used in the experiments and the actual \preach\ model file for the
prostate cancer therapy can be found in Appendix \ref{apndx:Models}.

All experiments were validated using Monte Carlo probability estimation in MATLAB (the
reported CPU times are for one core).
In particular, we calculated confidence intervals using the 
Chernoff-Hoeffding bound \cite{hoeffding}.
All results are given in the tables below, where the top half of each table contains the results 
obtained using our approach (\preach), while the bottom half reports the Monte Carlo results. 
For space reasons, the results of the bouncing ball are presented in the Appendix \ref{apndx:Models}.
Monte Carlo simulation of continuous nondeterminism (for \textbf{NPHA} models)
was implemented by first uniformly discretising the domain of the nondeterministic 
parameters. Then, for each (discretised) 
value of the parameters we built a confidence interval using the Chernoff-Hoeffding bound. 
Finally, the Monte Carlo interval reported in the tables below is the union of all such confidence 
intervals. Note that Monte Carlo intervals for \textbf{NPHA}s will be in general larger than $2\zeta$, and 
they will not be proper confidence intervals because of the nondeterministic parameters.
From the results we can see that our technique performs well even on highly nonlinear ODEs models such
as the prostate cancer treatment model, despite having unavoidably high complexity (see 
Theorem \ref{thm:complexity}). All Monte Carlo 
intervals cover the enclosures computed by \preach, thus confirming the correctness of our algorithms 
and their implementation. 

\vspace{-1ex}

\paragraph{\bf Table legend:}
{\small
$k$ = number of discrete transitions; $\epsilon$ = desired size of probability interval 
(\textbf{PHA}s only; lower box size limit for \textbf{NPHA});
$length$ = length of probability interval returned by \preach;
$\zeta,c$ = half-interval width and coverage probability for Chernoff bound;
$N$ = sample size from Chernoff bound;
$CPU$ = CPU time (sec).
}

\vspace{-3ex}

\begin{table}[ht!] 
\begin{adjustwidth}{-1in}{-1in}
\caption{Starvation model; see legend in Section \ref{sec:Exp}.}
\centering
\vspace{-1ex}
\begin{tabular}{c c c c c c c r}
\hline
Method&Model type&$k$ & $\epsilon$ & $length$ & Probability interval & $CPU$ &\\ [0.5ex] 

\hline 
\multirow{2}{4em}{{\sffamily Prob Reach}}&\textbf{NPHA}&0 & $10^{-3}$ & $4.245\cdot10^{-3}$ &  [0.9219413, 0.92618671] & 1,152 & \\ [0.5ex]
&\textbf{PHA}&0 & $10^{-3}$ & $6.795\cdot10^{-4}$ &  [0.92455817, 0.92523768] & 23 & \\ [0.5ex]
\hline\hline  
Method &Model type&$k$ & $\zeta$ & $c$ &  Monte Carlo interval & $CPU$ & $N$\\ [0.5ex] 
\hline 
\multirow{2}{4em}{Monte Carlo}&\textbf{NPHA}& 0 & $5\cdot 10^{-3}$ & 0.99 &  [0.9179, 0.9311] & 12,433 & 92,104\\ [0.5ex]
&\textbf{PHA}& 0 & $5\cdot10^{-3}$ & 0.99 &  [0.9193355, 0.9293355] & 2,868 & 92,104 \\ [0.5ex] 
\hline

\end{tabular} 
\label{table:starvation} 
\end{adjustwidth}
\end{table} 

\vspace{-8ex}

\begin{table}[ht!] 
\begin{adjustwidth}{-1in}{-1in}
\caption{Prostate cancer therapy model; see legend in Section \ref{sec:Exp}.}
\centering
\vspace{-1ex}
\begin{tabular}{c c c c c c c r}
\hline
Method&Model type&$k$ & $\epsilon$ & $length$ & Probability interval & $CPU$ &\\ [0.5ex] 

\hline 
\multirow{2}{4em}{{\sffamily Prob Reach}}&\textbf{PHA}&1 & $10^{-3}$ & $6.022\cdot10^{-4}$ & [0.47380981, 0.47441201] & 737 & \\ [0.5ex]
&\textbf{NPHA}&1 & $10^{-4}$ & $1.763\cdot10^{-3}$ & [0.4725522, 0.47431526] & 89,925 & \\ [0.5ex]
\hline\hline  
Method& Model type &$k$ & $\zeta$ & $c$ & Monte Carlo interval & $CPU$ & $N$\\ [0.5ex] 
\hline 
\multirow{2}{4em}{Monte Carlo}&\textbf{PHA}&1 & $1\cdot10^{-2}$ & 0.99 & [0.4648111, 0.4848111] & 5,700 & 23,026 \\ [0.5ex] 
&\textbf{NPHA}&1 & $1\cdot 10^{-2}$ & 0.99 & [0.4583, 0.4890] & 12,309 & 23,026\\ [0.5ex]
\hline

\end{tabular} 
\label{table:cancer} 
\end{adjustwidth}
\end{table} 

\vspace{-2ex}
\begin{table}[ht!] 
\begin{adjustwidth}{-1in}{-1in}
\caption{Car collision model; see legend in Section \ref{sec:Exp}.}
\centering
\vspace{-1ex}
\begin{tabular}{c c c c c c c r}
\hline
Method&Model type&$k$ & $\epsilon$ & $length$ & Probability interval & $CPU$ &\\ [0.5ex] 

\hline 
{\preach}&\textbf{PHA}&4 & $10^{-3}$ & $8.369\cdot10^{-4}$ & [0.5063922, 0.5072291] & 1,869 & \\ [0.5ex]
\hline\hline  
Method &Model type&$k$ & $\zeta$ & $c$ & Monte Carlo interval & $CPU$ & $N$\\ [0.5ex] 
\hline 
Monte Carlo&\textbf{PHA}&4 & $5\cdot 10^{-3}$ & 0.99 & [0.496629, 0.506629] & 32,201 & 92,104\\ [0.5ex]
\hline

\end{tabular} 
\label{table:cars} 
\end{adjustwidth}
\end{table}

%% file: appendix.tex
\section{Proofs}\label{apndx:proofs}

\begin{proof}[Proposition \ref{prop:BBorel}]
Immediate from the fact that (Definition \ref{def:HS}) the sets defined
by $\text{flow}_q, \text{init}_q, \text{jump}_q$, and $\text{unsafe}_q$ are Borel, and conjunction and
disjunctions correspond to set intersection and union, respectively.
\qed
\end{proof}

\begin{proof}[Proposition \ref{prop:IntegrationComplexity}]
It was proven in \cite[Corollary 6.3]{kobook} that the complexity of computing derivatives of 
$f\in P_{C^5 [a,b]}$ is P. Thus, computing a partial sum on an interval $[x]$ and 
evaluating the formula
\begin{equation*}
width([K]([x])) > \epsilon \frac{width([x])}{b - a}
\end{equation*}
is also polynomial in time. Given an arbitrary partition containing $n$ intervals, the formula above
can thus be verified in polynomial time with respect to the size of the partition. Hence, obtaining a partition such that on each interval the formula above holds is in NP complexity class.
\qed
\end{proof}

We recall that $P_{C^n[a,b]}$ denotes the class of polynomial-time (Type 2) computable functions
whose derivative $f^{(n)}$ exists and is continuous over $[a,b]$. Ko \cite[Section 6.2]{kobook}
showed that if $f$ is also analytic, then integration becomes P. However, such an algorithm
essentially uses truncated Taylor series over an arbitrary partition. Instead, our
Algorithm \ref{alg:integration} adaptively searches for a partition that guarantees the
required error bound $\epsilon$, while having a minimal number of intervals. In practice, this
significantly benefits the performance of our whole implementation. Another advantage of
Algorithm \ref{alg:integration} is that it does not require $f$ to be analytic.

\begin{proof}[Lemma \ref{lemma:converge}]
It was proven in \cite{DBLP:conf/lics/GaoAC12} that satisfiability of a first-order formula implies satisfiability of its weakening. Therefore, following can be equivalently derived:
\begin{equation*}
\phi \rightarrow \phi^{\delta'} \Leftrightarrow \neg \phi \vee \phi^{\delta'} \Leftrightarrow \phi^{\delta'} \vee \neg \phi \Leftrightarrow \neg \phi^{\delta'} \rightarrow \neg \phi
\end{equation*}

Let now $\psi = \phi^{\delta'}$ and $\psi^{\delta^{*}} = \phi^{\delta}$ be weakening of $\psi$. It was proven that if the weakening of the formula is unsatisfiable then the formula is also unsatisfiable. Then:
\begin{equation*}
\neg \psi^{\delta^{*}} \rightarrow \neg \psi \Leftrightarrow \neg \phi^{\delta} \rightarrow \neg \phi^{\delta'}
\end{equation*}
\qed
\end{proof}

\begin{proof}[Proposition \ref{prop:converge}]
By the definition of the decision procedure both formulas can be $\delta$-sat on an interval if and only if the considered interval contains values from the Borel set $B$ and its complement, or when a {\em false alarm} occurs. Then it can be concluded that the initial interval contains a subinterval which is either in the Borel set $B$ or outside it. This can be stated as:
\begin{equation} \label{eq:exists-interval}
\begin{split}
&\exists \delta \in \mathbb{Q}^{+}:(\phi([u, v])-\delta\textbf{-sat}) \wedge (\phi^{C}([u, v])-\delta\textbf{-sat}) \Rightarrow \\ 
&\exists [u', v'] \subseteq [u, v] : ([u', v'] \cap B = [u', v']) \oplus ([u', v'] \cap B = \emptyset)
\end{split}
\end{equation}

Then applying the decision procedure to $[u', v']$ and decreasing $\delta$, it is guaranteed that eventually we will obtain such a $\delta$ that the weakening of the formula will be false. In other words:
\begin{equation} \label{eq:exists-delta-unsat}
\exists \delta \in \mathbb{Q}^{+}: (\neg (\phi([u', v'])-\delta\textbf{-sat})) \oplus (\neg (\phi^{C}([u', v'])-\delta\textbf{-sat}))
\end{equation}

Therefore, by Lemma \ref{lemma:converge} the decision procedure will return \textbf{unsat} for one of the formulas in (\ref{eq:exists-delta-unsat}):
\begin{equation} \label{eq:exists-unsat}
\begin{split}
&\exists \delta \in \mathbb{Q}^{+}: (\neg (\phi([u', v'])-\delta\textbf{-sat})) \oplus (\neg (\phi^{C}([u', v'])-\delta\textbf{-sat})) \Rightarrow \\
&(\phi([u', v'])-\textbf{unsat}) \oplus (\phi^{C}([u', v'])-\textbf{unsat})
\end{split}
\end{equation}
\qed
\end{proof}

\begin{proof}[Proposition \ref{prop:multrv}]

If a hybrid system has $l$ independent initial random parameters with bounded support, then 
the reachability probability can be computed as:
\begin{equation} \label{eq:mulvarind}
  \int_{B} \prod_{i = 1}^{l} dP_i(r_i) = \int_{\Omega} I_{B}(r_{1}, ..., r_{l}) \prod_{i = 1}^{l} dP_i(r_{i})
\end{equation}
where $P_i$ is the probability measure of the $i$-th random parameter $r_i$, $B$ is the 
Borel set (\ref{def:Borel-set}) that contains all the random parameters values for which 
the hybrid system reaches the unsafe region in $k$ steps, $\Omega$ is the domain of the random parameters,
and $I_{B}(r_{1}, ..., r_{l})$ is the indicator function.

In order to compute (\ref{eq:mulvarind}) with precision $\epsilon_{prod}$, we must be able to compute 
\begin{equation} \label{eq:mulvarintegral}
\int_{\Omega} \prod_{i = 1}^{l} dP_i(r_{i})
\end{equation}
with the same precision. By Fubini's theorem, integral (\ref{eq:mulvarintegral}) can be calculated 
as the product 
\[
\int_{\Omega} \prod_{i = 1}^{l} dP_i(r_{i}) = \prod_{i = 1}^{l} \int_{a_{i}}^{b_{i}} dP_i(r_{i}) 
= \prod_{i = 1}^{l} I_{i}
\]
where 
\[
 I_i = \int_{a_{i}}^{b_{i}} dP_i(r_{i}) 
\]
and $a_i, b_i$ are the domain bounds of random parameter $r_i$.

Now, we can compute an interval of length $\epsilon_{i}$ containing the {\em exact} value of each 
integral $I_i$, and let us denote such interval as $[\hat{I}_i, \hat{I}_i +\epsilon_i]$. 
It is thus sufficient to demonstrate how the values $\epsilon_{i}$'s should be chosen 
in order for the integral (\ref{eq:mulvarintegral}) to be contained in an interval of length 
$\epsilon_{prod}$.

According to the rules of interval arithmetics, product of the intervals is contained in the interval:
\begin{equation}\label{eq:intervalprod}
[\hat{I_{1}}, \hat{I_{1}} + \epsilon_{1}] \cdot [\hat{I_{2}}, \hat{I_{2}} + \epsilon_{2}] \cdots 
[\hat{I_{l}}, \hat{I_{l}} + \epsilon_{l}] 
\ \subseteq\  [\prod_{i = 1}^{l} \hat{I_{i}}, \prod_{i = 1}^{l} (\hat{I_{i}} + \epsilon_{i})]
\end{equation}

Therefore, the $\epsilon_{i}$'s should be chosen such that the interval at the RHS 
of inclusion (\ref{eq:intervalprod}) has length smaller than $\epsilon_{prod}$, \ie, the following should hold:
\begin{equation} \label{eq:mulvarinequality}
\prod_{i = 1}^{l} (\hat{I_{i}} + \epsilon_{i}) - \prod_{i = 1}^{l} \hat{I_{i}} \leq \epsilon_{prod} 
\end{equation}

Therefore, choosing $\epsilon_{i}$ in such a way that (\ref{eq:mulvarinequality}) holds will guarantee that the {\em exact} value of the product of $l$ integrals is contained in the interval of size $\epsilon_{prod}$. 
If we want all the $\epsilon_{i}$'s equal to a single value $\epsilon$, then formula 
(\ref{eq:mulvarinequality}) can be satisfied by assuming in the worst case $\hat{I}_i = 1$ for all $i$,
which gives
\[
\epsilon_{prod} \geq \prod_{i=1}^l (1+\epsilon) -1 =  \sum_{i=1}^l {{l}\choose{i}} \epsilon^i
\]
where ${{l}\choose{i}}$ is the binomial coefficient.

\qed
\end{proof}

\begin{proposition}\label{prop:bounded}
Given $\epsilon \in (0, 1] \cap \mathbb{Q}$, $k\in\mathbb{N}$ and a hybrid system with one
\textbf{bounded} continuous random initial parameter, there exists an algorithm for computing an interval
of size not larger than $\epsilon$ that contains the value of (\ref{eq:probabilistic-reachability}),
\ie, the probability of reaching the unsafe region in $k$ steps.
\end{proposition}
\begin{proof}
Let $r \in [a, b]$ be a random continuous parameter. Then by applying our validated integration procedure
(Algorithm \ref{alg:integration}) we obtain a partition $\cup_{i = 1}^{n} [r]_{i}$ such that on each of the intervals the value of the partial sum is enclosed by an interval of length $\epsilon \frac{width([r]_{i})}{width([a, b])}$, and the value of the integral on $[a, b]$ is enclosed by the interval of size $\epsilon$.

Let $k$-th step reachability be encoded by the formula $\phi$, and $\phi^{C}$ be derived as in (\ref{eq:phi_C}).
By applying the decision procedure to all the intervals from the initial partition, we can distributed them in three sets $B_{unsat}, B_{C^{unsat}}, B_{\delta-sat}$ containing the intervals where $\phi$ is \textbf{unsat}, $\phi^{C}$ is \textbf{unsat}, and both formulas are $\delta$\textbf{-sat}, respectively. Then the following will hold:
\begin{equation} \label{eq:a-b-integral}
\int_{a}^{b}\, f(r)\, dr = \int_{B_{unsat}}\, f(r)\, dr + \int_{B_{C^{unsat}}}\, f(r)\, dr + \int_{B_{\delta-sat}}\, f(r)\, dr
\end{equation}
The lower and the upper bounds of the interval containing the exact value of the probability can be found as:
\[
\begin{split}
&P_{lower} = \int_{B_{C^{unsat}}}\, f(r)\, dr \\
&P_{upper} = \int_{a}^{b}\, f(r)\, dr - \int_{B_{unsat}}\, f(r)\, dr
\end{split}
\]
Then size of the interval $[P_{lower}, P_{upper}]$ can be calculated as:
\[
\begin{split}
&P_{upper} - P_{lower} = \int_{a}^{b}\, f(r)\, dr - \int_{B_{unsat}}\, f(r)\, dr - \int_{B_{C^{unsat}}}\, f(r)\, dr = \\
&= \int_{B_{\delta-sat}}\, f(r)\, dr
\end{split}
\]

By Proposition \ref{prop:converge} it follows that on each interval $[u, v]$ in $B_{\delta-sat}$ we can obtain a subinterval $[u', v']$ such that it can be added to $B_{unsat}$ or $B_{C^{unsat}}$ and, thus, removed from $B_{\delta-sat}$. Therefore, as $\delta \rightarrow 0$ and $n \rightarrow \infty$ (where $n$ is the number of disjoint subintervals partitioning $B_{\delta-sat}$) the size of set $B_{\delta-sat}$ will be decreasing. Hence, we can conclude that $\int_{B_{\delta-sat}}\, f(r)\, dr \rightarrow 0$, which implies:
\[
\exists \epsilon \in \mathbb{Q}^{+}: P_{upper} - P_{lower} \le \epsilon
\]
\qed
\end{proof}

\begin{proposition} \label{prop:unbounded}
Given $\epsilon \in (0, 1] \cap \mathbb{Q}$, $k\in\mathbb{N}$ and a hybrid system with one
\textbf{unbounded} continuous random initial parameter, there exists an algorithm for computing
an interval of size not larger than $\epsilon$ that contains the value
of (\ref{eq:probabilistic-reachability}).
\end{proposition}
\begin{proof}
Let us recall that calculating the probability of reaching the unsafe region requires integrating an
indicator function with respect the probability measure associated to the random parameter
\begin{equation} \label{eq:probability}
\int_{\Omega}\, I_{B}(r)\, dP(r)
\end{equation}
where $I_{B}(r)$ is the indicator function over set $B$ (\ref{def:Borel-set}), $P$ is the probability 
measure of the random variable, and $\Omega = (-\infty, +\infty)$. In the following we shall simplify
notation and write $dP$ instead of $dP(r)$, since there is only one random variable.

The next inequality can be readily derived from the definition of indicator function:
\[
0 \le \int_{\Omega}\, I_{B}(r)\, dP \le \int_{\Omega}\, dP .
\]
By the property of definite integral, for any $a$ and $b \ge a$:
\[
\int_{\Omega}\, I_{B}(r)\, dP = \int_{a}^{b}\, I_{B}(r)\, dP + \int_{-\infty}^{a}\, I_{B}(r)\, dP + \int_{b}^{\infty}\, I_{B}(r)\, dP .
\]
As $\int_{-\infty}^{a}\, I_{B}(r)\, dP \ge 0$ and $\int_{b}^{\infty}\, I_{B}(r)\, dP \ge 0$, the following 
holds for all $r \in \Omega$
\[
\int_{a}^{b}\, I_{B}(r)\, dP \le \int_{\Omega}\, I_{B}(r)\, dP \le \int_{a}^{b}\, I_{B}(r)\, dP + 1 - \int_{a}^{b}\, dP .
\]
Therefore, the exact value of probability is enclosed by the interval:
\[
\int_{\Omega}\, I_{B}(r)\, dP  \in [\int_{a}^{b}\, I_{B}(r)\, dP, \int_{a}^{b}\, I_{B}(r)\, dP + 1 - \int_{a}^{b}\, dP]
\]

By Proposition \ref{prop:bounded} we can calculate the lower and the upper bounds of the probability 
over the bounded interval $[a, b]$:
\begin{equation} \label{eq:exact-probability}
\int_{\Omega}\, I_{B}(r)\, dP  \in [(\int_{a}^{b}\, I_{B}(r)\, dP)_{lower}, (\int_{a}^{b}\, I_{B}(r)\, dP)_{upper} + 1 - \int_{a}^{b}\, dP]
\end{equation}

Now it is desired that the interval in formula (\ref{eq:exact-probability}) is of length $\epsilon$. For this the error $\epsilon$ can be presented as a sum of two components $\epsilon_{inf}$ and $\epsilon_{prob}$ that are chosen such that: $\epsilon \ge \epsilon_{inf} + \epsilon_{prob}$ where $\epsilon_{inf} \ge 1 - \int_{a}^{b}\, dP$ and $\epsilon_{prob} \ge (\int_{a}^{b}\, I_{B}(r)\, dP)_{upper} - (\int_{a}^{b}\, I_{B}(r)\, dP)_{lower}$.

The values $a$ and $b$ can be obtained by solving the first inequality as a first order formula:
\begin{equation} \label{eq:epsilon-inf}
\exists a \in [u_{a}, v_{a}], \exists b \in [u_{b}, v_{b}] : (\frac{dF}{dx} = f(x)) \wedge (F(a) = 0) \wedge (F(b) \ge 1 - \epsilon_{inf})
\end{equation}
where $f$ is the probability density function of the random parameter, which is known to the user. (Note that
$F$ thus denotes the cumulative distribution function of the random parameter.)
Then the values $a$ and $b$ derived from formula (\ref{eq:epsilon-inf}) are used to compute the interval $[(\int_{a}^{b}\, I_{B}(r)\, dP)_{lower}, (\int_{a}^{b}\, I_{B}(r)\, dP)_{upper}]$ of length $\epsilon_{prob}$. This can be performed for an arbitrary positive rational number (by Proposition \ref{prop:bounded}).

If formula (\ref{eq:epsilon-inf}) is unsatisfiable then it means that bounds for the variables $a$ and $b$ should be enlarged and the formula should be verified again. This process should repeat until the formula is satisfiable and the values $a$ and $b$ are obtained.

\qed
\end{proof}

\begin{proof}[Proposition \ref{prop:cn_nd}]
Evaluating formulas $\phi$ and $\phi^C$ on two boxes $[\overline{r}]$ and $[\overline{z}]$ (over random and nondeterministic continuous parameters, respectively) there are four possible outcomes:
\begin{itemize}
        \item{$\phi([\overline{r}], [\overline{z}])$ is \textbf{unsat}. Hence, there are {\em for sure} no values in $[\overline{r}]$ and $[\overline{z}]$ such that the system reaches the unsafe region, so $[\overline{r}]$ is not in $B$.}
        \item{$\phi([\overline{r}], [\overline{z}])$ is \textbf{$\delta$-sat}. Then, there is a value in $[\overline{r}], [\overline{z}]$ such that the system reaches $U$ or $U^{\delta}$ ($\delta$-weakening of set $U$).}
        \item{$\phi^C([\overline{r}], [\overline{z}])$ is \textbf{unsat}. Therefore, there is {\em for sure} no value in $[\overline{r}]$ and $[\overline{z}]$ such that for all time points on the $k$-th step the system stays within the complement of the unsafe region. In other words, for all the values in $[\textbf{r}]$ the system    reaches the unsafe region, so $[\textbf{r}]$ is fully contained in $B$.}
        \item{$\phi^C([\overline{r}], [\overline{z}])$ is \textbf{$\delta$-sat}. Then there is a value in $[\overline{r}], [\overline{z}]$ such that the system stays within $U^{C}$ or $U^{C^{\delta}}$.}
\end{itemize}
Similarly to the approach used in the proof of Proposition \ref{prop:bounded}, lower and upper bounds of 
the reachability probability can be calculated as:
\[
\begin{split}
&P_{lower} = \int_{B_{C^{unsat}}}\, f(r)\, dr \\
&P_{upper} = \int_{a}^{b}\, f(r)\, dr - \int_{B_{unsat}}\, f(r)\, dr
\end{split}
\]
where $B_{unsat}, B_{C^{unsat}}, B_{\delta-sat}$ containing the boxes where $\phi$ is \textbf{unsat}, $\phi^{C}$ is \textbf{unsat}, and both formulas are $\delta$\textbf{-sat} respectively.
Hence, by refining boxes from $B_{\delta-sat}$ until $max(|[\textbf{r}]|) \le \epsilon$, 
we obtain an interval $[P_{lower}, P_{upper}]$ containing the range of probabilities of reaching 
the unsafe region.
\qed
\end{proof}

\begin{proof}[Theorem \ref{thm:dd}]
Let $\phi$ be a formula describing a hybrid system with discrete random parameters 
$\{{D}_{1}, \ldots, {D}_{p}\}$, and let $p(\cdot)$ denotes their probabilities, \ie,
for each $D_{i}$ we have that $ \sum_{j=1}^{\#D_{i}} p(d_{ij}) = 1$.

Let $\textbf{DD} = {D}_{1} \times \cdots \times {D}_{p}$ be the Cartesian product of discrete parameters.
For each $\textbf{dd} = \{d_{11}, d_{22}, ..., d_{pk}\} \in \textbf{DD}$, 
let $m_{\textbf{dd}} = p(d_{11}) \cdot p(d_{21}) \cdot \cdots \cdot p(d_{pk})$.
Substituting all discrete random parameters with their values from \textbf{dd} we will obtain a hybrid system which can be described by a corresponding formula $\phi_{\textbf{dd}}$.

Now depending of the type of the considered hybrid system we can use one of the algorithms already presented.
\begin{itemize}
        \item{\textbf{PHA} or \textbf{NPHA}: we can apply Algorithm \ref{alg:prob-reach} and obtain a probability interval $[\underline{[P_{lower}]}, \overline{[P_{upper}]}]_{\textbf{dd}}$.}
        \item{\textbf{HA}: we can just use the decision procedure described above and evaluate $\phi_{\textbf{dd}}$ and $\phi^{C}_{\textbf{dd}}$. Then returned value depends on the evaluation outcome:
                \begin{itemize}
                        \item{$\phi_{\textbf{dd}}$-\textbf{unsat} return $[\underline{[P_{lower}]}, \overline{[P_{upper}]}]_{\textbf{dd}} = [0.0, 0.0]$}
                        \item{$\phi^{C}_{\textbf{dd}}$-\textbf{unsat} return $[\underline{[P_{lower}]}, \overline{[P_{upper}]}]_{\textbf{dd}} = [1.0, 1.0]$}
                        \item{$\phi_{\textbf{dd}}$-\textbf{$\delta-sat$} and $\phi^{C}_{\textbf{dd}}$-\textbf{$\delta-sat$} return $[\underline{[P_{lower}]}, \overline{[P_{upper}]}]_{\textbf{dd}} = [0.0, 1.0]$}
                \end{itemize}}
\end{itemize}

Doing so for each $\textbf{dd} \in \textbf{DD}$ we can obtain the resulting probability interval
\begin{equation*}
[P] = \sum_{\textbf{dd} \in \textbf{DD}}(m_{\textbf{dd}} \cdot [\underline{[P_{lower}]}, \overline{[P_{upper}]}]_{\textbf{dd}})
\end{equation*}
\qed
\end{proof}

\begin{proof}[Theorem \ref{thm:complexity}]
The considered algorithm can be presented as two independent components: validated integration and probability calculation.

The decision procedure used in the algorithm consists of two formulas: $\phi$ and $\phi^{C}$,
which are $\Sigma_{1}$ and $\Sigma_{2}$ sentences. Solving these formulas as a $\delta$-SMT problem
is in $(\Sigma_{1}^{P})^{C}$ and $(\Sigma_{2}^{P})^{C}$ complexity classes respectively,
where $P \subseteq C \subseteq PSPACE$ is the complexity of the terms in the 
formula \cite{DBLP:conf/lics/GaoAC12}. 
Hence, the decision procedure is in $(\Sigma_{2}^{P})^{C}$.
Then verification of an arbitrary partition of $n$ intervals is also in $(\Sigma_{2}^{P})^{C}$.
Hence, it is clear that obtaining the {\em correct} partition (such that 
$\overline{[P_{upper}]} - \underline{[P_{lower}]} \le \epsilon_{prob}$) is in 
$NP^{(\Sigma_{2}^{P})^{C}}$.
By Proposition \ref{prop:IntegrationComplexity} the complexity of the verified integration is $NP$,
which is in $NP^{(\Sigma_{2}^{P})^{C}}$. The complexity of the whole algorithm is thus 
$NP^{(\Sigma_{2}^{P})^{C}}$, where $P \subseteq C \subseteq PSPACE$.

Finally, it has been shown in \cite{DBLP:conf/lics/GaoAC12} that if $\phi$ (and thus $\phi^C$) 
includes Lipschitz-continuous ODEs then the $\delta$-SMT problem becomes $PSPACE$-complete.
This lifts the complexity of the whole algorithm to $PSPACE$-complete.

\qed
\end{proof}

\section{Models}\label{apndx:Models}
We give here more information about the models used for our experiments.

\subsection{2D-moving Bouncing ball}
The ball is launched from position $(S_{x} \in [-5, 5], S_{y} = 0)$ with initial speed 
$\upsilon_0 \sim N(20,1)$, \ie, normal distribution with mean 20 and variance 1, and angle 
$\alpha$ to horizon (measured in radians) with the following probability distribution: 
\[
\begin{array}{l}
P[\alpha = 0.5236] = 0.3\\
P[\alpha = 0.7854] = 0.5\\
P[\alpha = 1.0472] = 0.2
\end{array}
\]
After each jump the speed of the ball is multiplied by 0.9. The gravity of Earth parameter
$g\in[9.8,9.81]$ is also nondeterministic. The system is modelled as a hybrid system with 
one mode with dynamics governed by a system of ODEs: 
\[
\begin{array}{l}
S'_{x}(t) = \upsilon_0 \cos{\alpha} \\[1ex]
S'_{y}(t) = \upsilon_0 \sin{\alpha} - gt
\end{array}
\]
The goal of the experiment is to calculate the probability of reaching the region 
$S_{x}(t) \ge 100$ within 0 and 1 jump. The results are presented in Table \ref{table:bouncing-ball}. 
Monte Carlo simulation of continuous nondeterminism in MATLAB was achieved as explained
in Section \ref{sec:Exp}, using uniform discretisation of the domains of the nondeterministic 
parameters ($S_{x}$ and $g$ were discretised with 100 and 10 values, respectively). 
In Figure \ref{fig:bb-graph0} and \ref{fig:bb-graph1} we plot the Monte Carlo reachability 
probability estimate with respect to the nondeterministic parameters $S_x(0)$ and $g$, for 0 and 1 jump,
respectively.
\begin{figure}[ht!] 
	\centering
	\includegraphics[width=120mm]{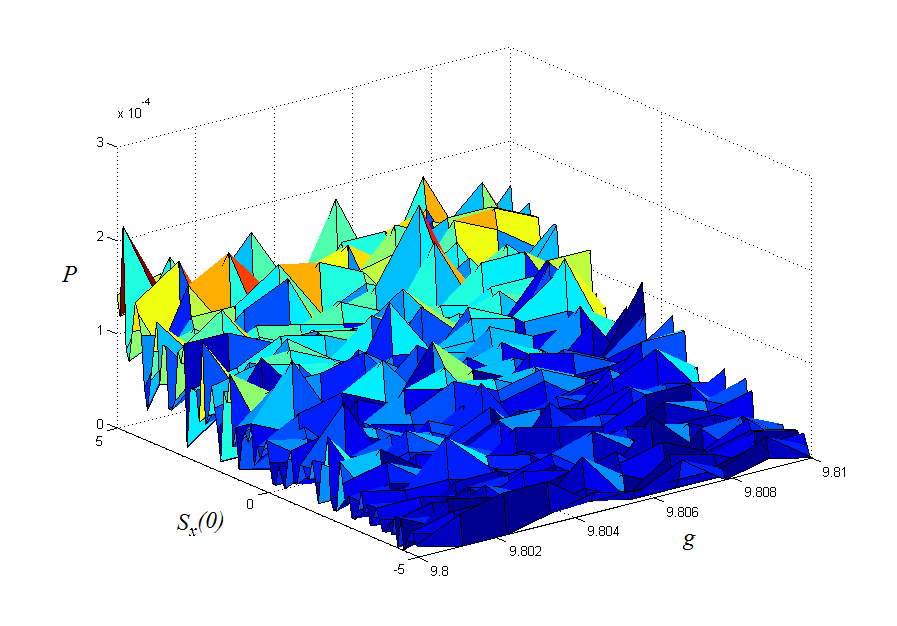}
	\caption{Monte Carlo simulation of the bouncing ball model: Reachability probability ($k=0$) 
	estimate $P$ with respect to nondeterministic parameters $S_{x}(0)\in [-5,5]$ and $g\in [9.8,9.81]$.}
	\label{fig:bb-graph0}
\end{figure}

\begin{figure}[ht!] 
	\centering
	\includegraphics[width=120mm]{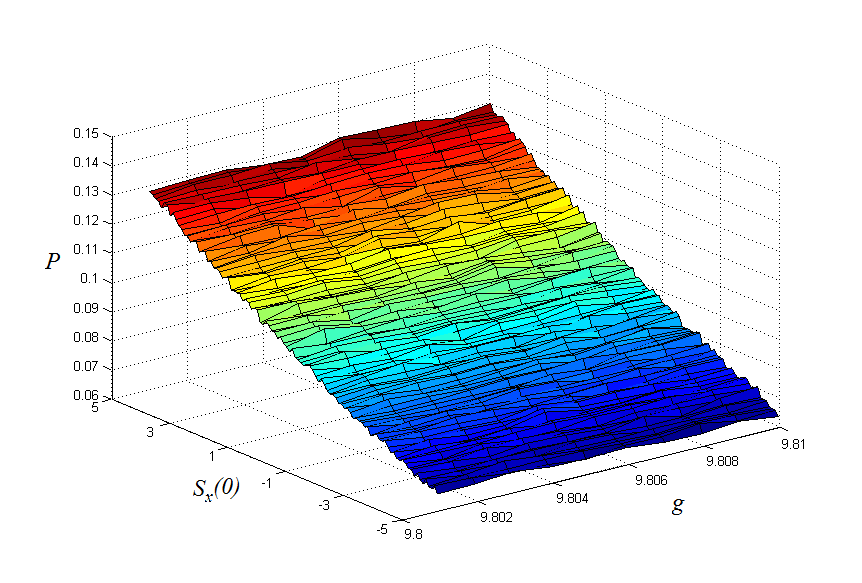}
	\caption{Monte Carlo simulation of the bouncing ball model: Reachability probability ($k=1$) 
	estimate $P$ with respect to nondeterministic parameters $S_{x}(0)\in [-5,5]$ and $g\in [9.8,9.81]$.}
	\label{fig:bb-graph1}
\end{figure}

\begin{figure}[ht!] 
	\centering
	\includegraphics[width=90mm]{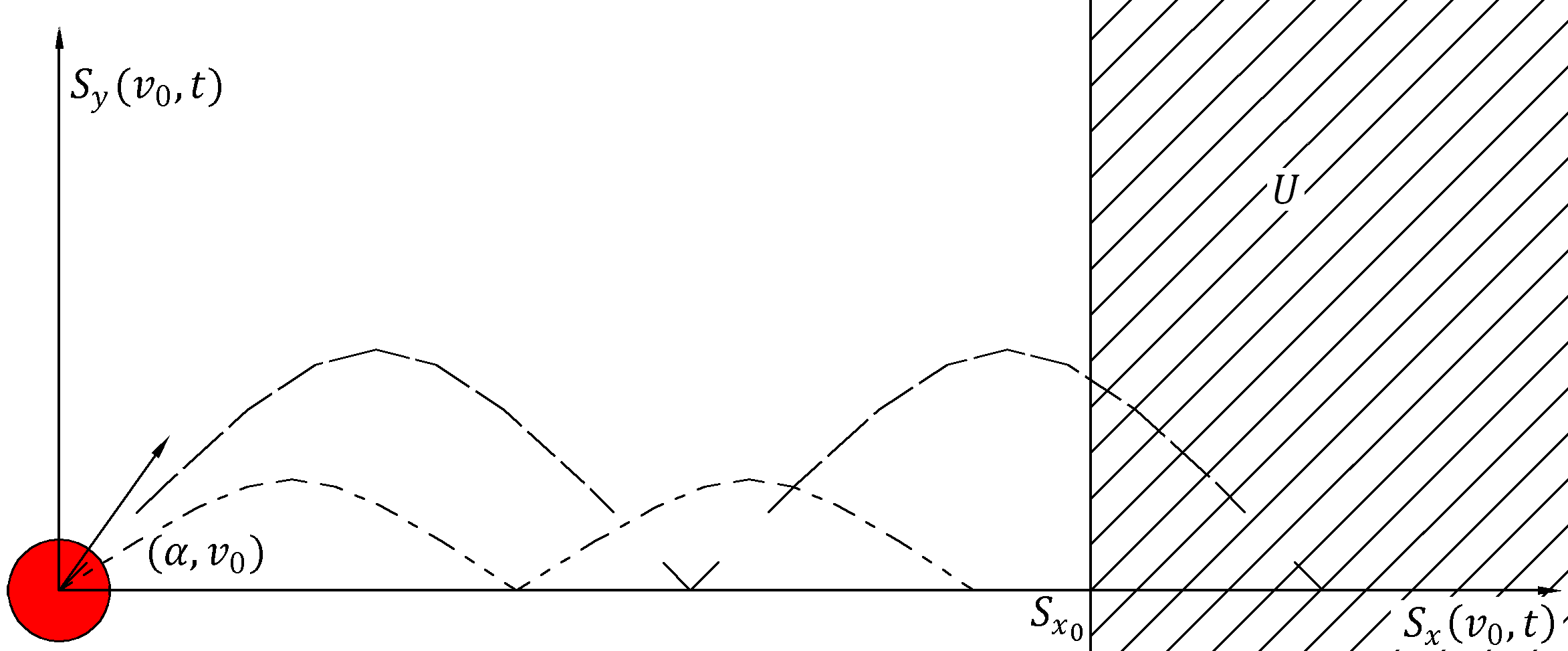}
	\caption{2D-moving bouncing ball scenario}
	\label{fig:bouncing-ball}
\end{figure}

\begin{table}[ht!]
\begin{adjustwidth}{-1in}{-1in}
\centering
\caption{2D-moving bouncing ball model; see legend in Section \ref{sec:Exp}.}
\vspace{-1ex}
\begin{tabular}{c c c c c c c r}
\hline
Method & Model type& $k$ & $\epsilon$ & $length$ & Probability interval & $CPU$ & \\ [0.5ex]
\hline
\multirow{2}{4em}{{\sffamily Prob Reach}} &\textbf{NPHA}& 0 & $10^{-3}$ & $2.38\cdot10^{-4}$ & [0.000013103, 0.000250681] & 223 & \\ [0.5ex]
&\textbf{NPHA}& 1 & $10^{-3}$ & $6.464\cdot10^{-2}$ &[0.0647381, 0.12937951] & 1,605 & \\ [0.5ex]
\hline \hline
Method & Model type & $k$ & $\zeta$ & $c$ & Monte Carlo interval & $CPU$ & $N$\\ [0.5ex]
\hline
\multirow{2}{4em}{Monte Carlo}&\textbf{NPHA}& 0 & $5\cdot 10^{-3}$ & 0.99 & [0, 0.00520629] & 1,482 & 92,104\\ [0.5ex]
&\textbf{NPHA}& 1 & $5\cdot10^{-3}$ & 0.99 & [0.0585, 0.1367] & 1,485 & 92,104 \\ [0.5ex]
\hline
\end{tabular}
\label{table:bouncing-ball}
\end{adjustwidth}
\end{table}

\subsection{Starvation model}
In humans, enduring fasting for 3-4 days will consume all the glucose reserves of the body.
At this point, the energy to sustain the human body is produced from fat $F(t)$, muscles $M(t)$ 
and ketone bodies $K(t)$ (for brain function) \cite{Starvation}. The ODE system below represents 
the dynamics of the described variables:
\begin{equation*}
\begin{split}
\frac{dF}{dt} &= F(\frac{- a}{1 + K} - \frac{1}{\lambda_{F}}  (\frac{C + g L_{0}}{F + M} + g)) \\
\frac{dM}{dt} &= - \frac{M}{\lambda_{M}}  (\frac{C + g L_{0}}{F + M} + g) \\
\frac{dK}{dt} &= \frac{V a F}{1 + K} - b
\end{split}
\end{equation*}
We consider two scenarios where parameter $g\sim N(10.96,1)$, \ie, normally distributed with
mean 10.96 and variance 1, and:
\begin{itemize}
	\item{$b \in [0.05, 0.075]$ is nondeterministic}; or
	\item{$b$ is a discrete random parameter with the probability distribution: }
\[
\begin{array}{l}
P[b = 0.05] = 0.1\\
P[b = 0.06] = 0.2\\
P[b = 0.07] = 0.3\\
P[b = 0.075] = 0.4
\end{array}
\]
\end{itemize}

The probabilistic reachability property investigated in the experiment is: {\em what is 
the probability that muscle mass will decrease by 40\% within 25 days?}
Numerical values for all deterministic parameters in the model are presented 
in Table \ref{table:starvation-param} and verification results are featured in 
Table \ref{table:starvation}. Monte Carlo simulation of continuous nondeterminism 
in MATLAB was achieved as explained in Section \ref{sec:Exp}, via uniform discretisation 
(10 values) of the nondeterministic parameter $b$. In Figure \ref{fig:starvation-graph} we
plot the Monte Carlo reachability probability and confidence interval with respect to the 
value of parameter $b$.

\begin{table}[ht!] 
\caption{Starvation model parameters and initial conditions}
\centering
\begin{tabular}{c c c c c c c c c c}
\hline\hline 
Param. & Value & Param. & Value & Param. & Value & Param. & Value & Param. & Value\\ [0.5ex] 

\hline 
$a$ & 0.013 & $\lambda_{F}$ & 7777.8 & $M(0)$ & 43.6 & $V$ & 0.9 & $K(0)$ & 0.02\\ [0.5ex]
$C$ & 772.3 & $\lambda_{M}$ & 1400 & $F(0)$ & 25 & $L_{0}$ & 30.4 &  & \\ [0.5ex]
\hline \\ [0.5ex]
\end{tabular} 
\label{table:starvation-param} 

\end{table}

\begin{figure}[ht!] 
\centering
\includegraphics[width=120mm]{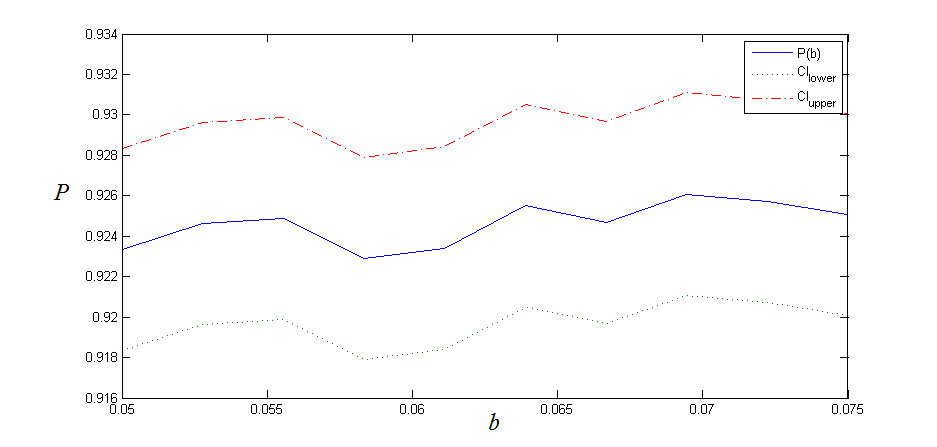}
\caption{Monte Carlo simulation of the starvation model: Reachability probability estimate $P$ 
(solid line) with respect to nondeterministic parameter $b\in [0.05, 0.075]$. For each 
(discretised) value of $b$  we give a Chernoff-Hoeffding confidence interval, denoted by dotted lines.}
\label{fig:starvation-graph}
\end{figure}

\subsection{Road scenario}

\begin{figure}[ht!] 
\centering
\includegraphics[width=120mm]{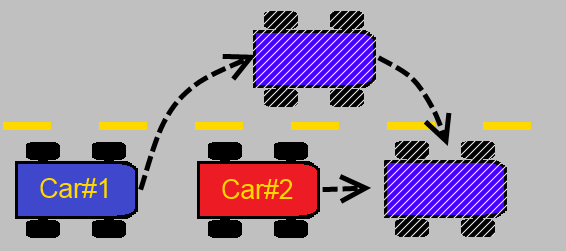}
\caption{Road scenario}
\label{fig:cars}
\end{figure}

We consider the road scenario inspired by a model presented in \cite{SMCnd} and depicted in 
Figure \ref{fig:cars}. Two cars ($Car\#1$ and $Car\#2$) move
on the same lane, starting at coordinates $S_{01} = 0$ and $S_{02} = S_{01} + \upsilon_{1} \cdot t_{safe}$,
where $t_{safe} = 2 sec$ implements the so-called ``two seconds rule'' for maintaining a safety 
distance between two cars.

We describe a car collision scenario and we model it with the hybrid automaton given in 
Figure \ref{fig:cars_model}. Starting in Mode 1 at time $t = 0$, $Car\#1$ changes lane and starts 
accelerating at $a_{a1}$ $m/s^{2}$, while $Car\#2$ is moving in the initial lane 
with speed $\upsilon_{2}$. Upon reaching the maximum speed $\upsilon_{max}$, the
system switches to Mode 2, where $Car\#1$ keeps moving at this speed until it gets ahead 
of $Car\#2$ by the safety distance $S_{safe} = \upsilon_{2} \cdot t_{safe}$.
After that, we switch to Mode 3: $Car\#1$ returns to the initial lane and starts decelerating 
at $a_{d1}$. For the driver of $Car\#2$ it takes $t_{react} = 1 sec$ to react
(the system switches to Mode 4) and then it starts decelerating as well (with random acceleration 
$a_{d2} \sim N(-1.35, 0.01)$). In Mode 3, 4, and 5 we also have an invariant specifying that $Car\#1$
should precede $Car\#2$ at all time. We calculate the {\em probability of observing a car collision 
in Mode 5}, where $Car\#1$ is stopped. Numerical values for all deterministic parameters in the model 
are given in Table \ref{table:cars-param} are verification results are presented in Table \ref{table:cars}.

\begin{table}[ht!] 
\caption{Car collision model parameters and initial conditions}
\centering
\begin{tabular}{c c c c c c c c c c}
\hline\hline 
Param. & Value & Param. & Value & Param. & Value & Param. & Value & Param. & Value\\ [0.5ex] 

\hline 
$\upsilon_{1}$ & 11.12 & $\upsilon_{2}$ & 11.12 & $\upsilon_{max}$ & 16.67 & $t_{safe}$ & 2 & $S_{safe}$ & $\upsilon_{2} \cdot t_{safe}$\\ [0.5ex]
$a_{a1}$ & 3 & $a_{d1}$ & -4 & $t_{react}$ & 1 & $S_{01}$ & 0 & $S_{02}$ & $S_{01} + \upsilon_{1} \cdot t_{safe}$ \\ [0.5ex]
\hline \\ [0.5ex]
\end{tabular} 
\label{table:cars-param} 

\end{table}

\begin{figure}[ht!] 
\centering
\includegraphics[width=120mm]{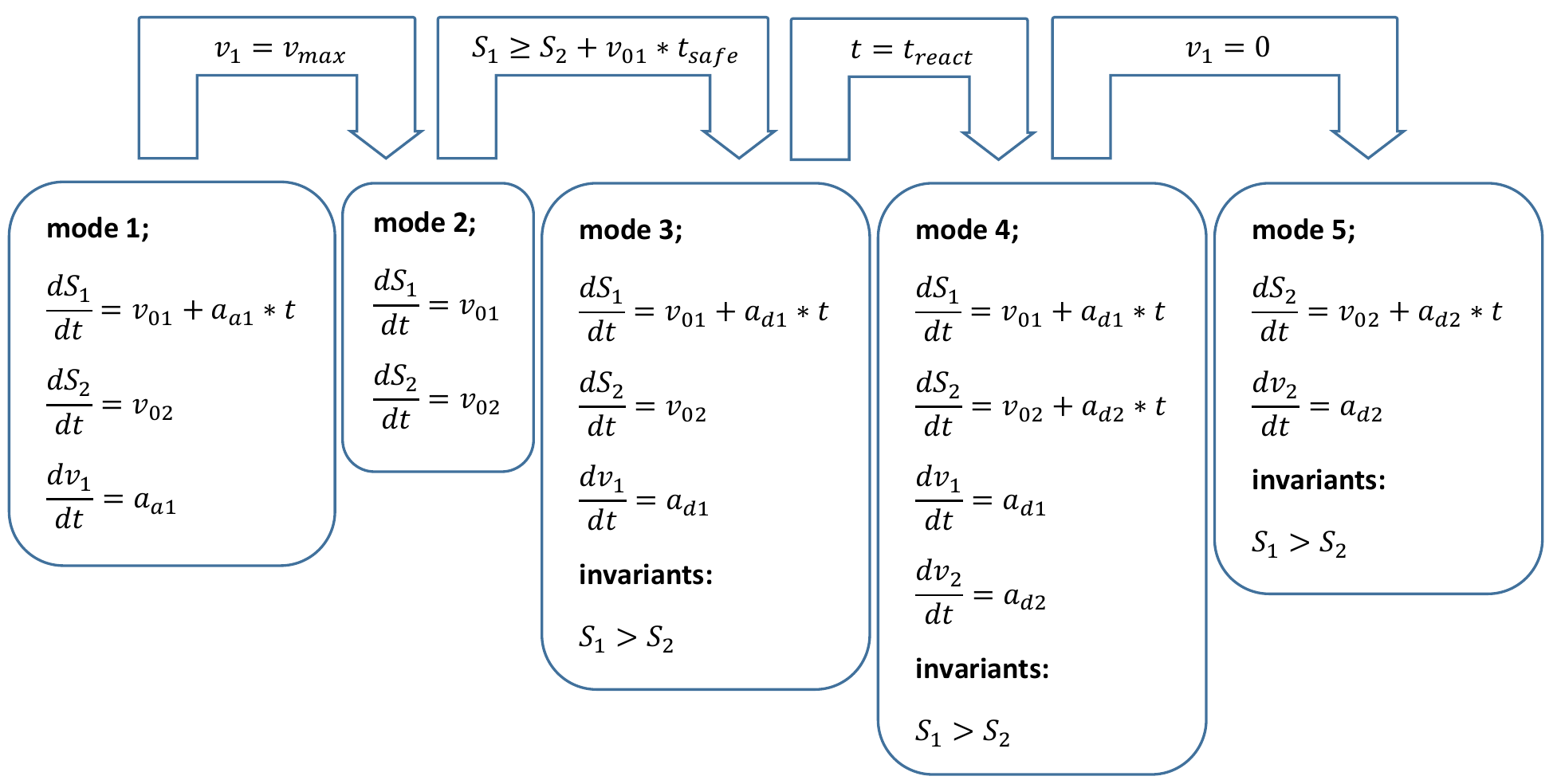}
\caption{Car collision model}
\label{fig:cars_model}
\end{figure}

\subsection{Prostate cancer therapy}
We consider a model of personalised prostate cancer therapy introduced by 
Ideta {\em et al.}~\cite{Ideta2008} and improved by Liu {\em et al.}~\cite{Liu2015}. 
Intermittent androgen suppression (IAS) has proved to be more effective than constant androgen 
suppression (CAS) in delaying the recurrence of prostate cancer. Briefly, the personalised 
therapy comprises of two repeating stages. The patient's prostate-specific antigen (PSA) level 
is monitored throughout the therapy. When the PSA level reaches an upper threshold, the patient 
starts receiving treatment ({\em on-therapy} stage) until the PSA level decreases to 
a lower threshold ({\em off-therapy}). The main aim of the therapy is to delay cancer relapse 
for as long as possible.

The model of the therapy is given in Figure \ref{fig:cancer_model} (a full explanation of the 
model and its parameters can be found in \cite{Liu2015}). Mode 1 is the {\em on-therapy} 
stage, and it continues until the PSA level (measured by $x+y$) is above threshold $r_{0} = 4$. 
Then the system makes a transition to the {\em off-therapy} mode which continues until PSA level 
is below $r_{1} = 10$. We explore the following scenarios:
\begin{itemize}
	\item{$\alpha_{y}$ is distributed normally ($N(0.05,0.01)$)} and $\alpha_{x}= 0.0197$
	\item{$\alpha_{y}$ is distributed normally ($N(0.05,0.01)$) and $\alpha_{x} \in [0.0197,0.0204]$ is nondeterministic}
\end{itemize}
For the cases above we calculate the probability of cancer relapse (\ie, $y\geq 1$) within 100 
days of using the personalised cancer therapy. Numerical values of all the parameters in the model 
are presented in the Table \ref{table:cancer-param} and verification results are featured in 
Table \ref{table:cancer}. Monte Carlo simulation of continuous nondeterminism in MATLAB 
was achieved as detailed in Section \ref{sec:Exp}, using uniform discretisation (20 values) of
the domain of the nondeterministic parameter $\alpha_{x}$. In Figure \ref{fig:cancer-graph} we
plot Monte Carlo reachability probability and confidence interval with respect to the value 
of parameter $\alpha_x$.

\begin{figure}[ht!] 
\centering
\includegraphics[width=120mm]{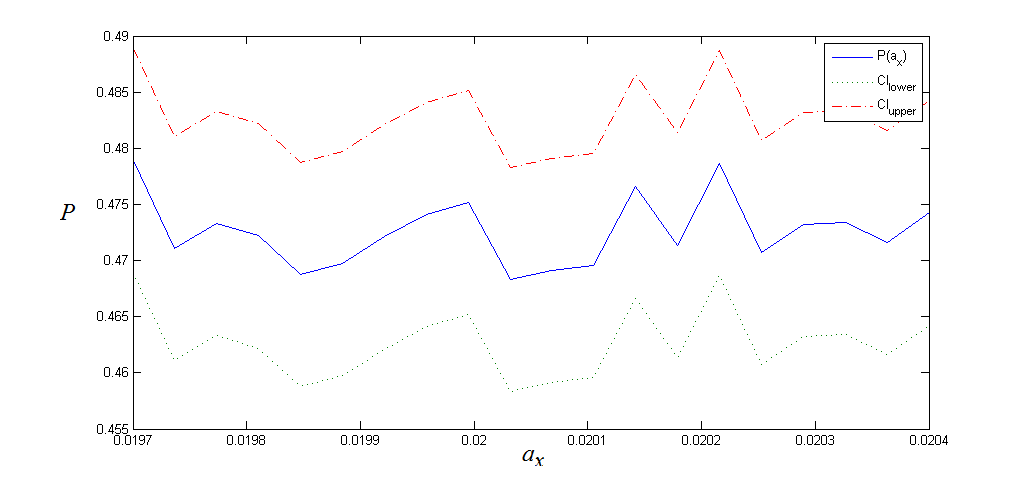}
\caption{Monte Carlo simulation of the prostate cancer therapy model: Reachability probability estimate
$P$ (solid line) with respect to the nondeterministic parameter $\alpha_x\in [0.0197,0.0204]$. 
For each (discretised) value of $\alpha_x$ we give a Chernoff-Hoeffding confidence interval, 
denoted by dotted lines.}
\label{fig:cancer-graph}
\end{figure}

\begin{table}[ht!] 
\caption{Prostate cancer therapy model parameters and initial conditions}
\centering
\begin{tabular}{c c c c c c c c c c}
\hline\hline 
Param. & Value & Param. & Value & Param. & Value & Param. & Value & Param. & Value\\ [0.5ex] 

\hline 
$\beta_{x}$ & 0.0175 & $\beta_{y}$ & 0.0168 & $k_1$ & 10.0 & $k_2$ & 1.0 & $k_3$ & 10.0\\ [0.5ex]
$k_4$ & 2 & $m_1$ & $10^{-5}$ & $z_0$ & 12 & $\gamma$ & 0.08 & $r_1$ & 10.0\\ [0.5ex]
$r_0$ & 4.0 & $d_0$ & $1.0$ & $c_1$ & 0.01 & $c_2$ & 0.03 & $c_3$ & 0.02\\ [0.5ex]
$x(0)$ & 19 & $y(0)$ & $0.1$ & $z(0)$ & 12.5 &  &  &  & \\ [0.5ex]
\hline \\ [0.5ex]
\end{tabular} 
\label{table:cancer-param} 
\end{table}

\begin{figure}[] 
\centering
\includegraphics[width=120mm]{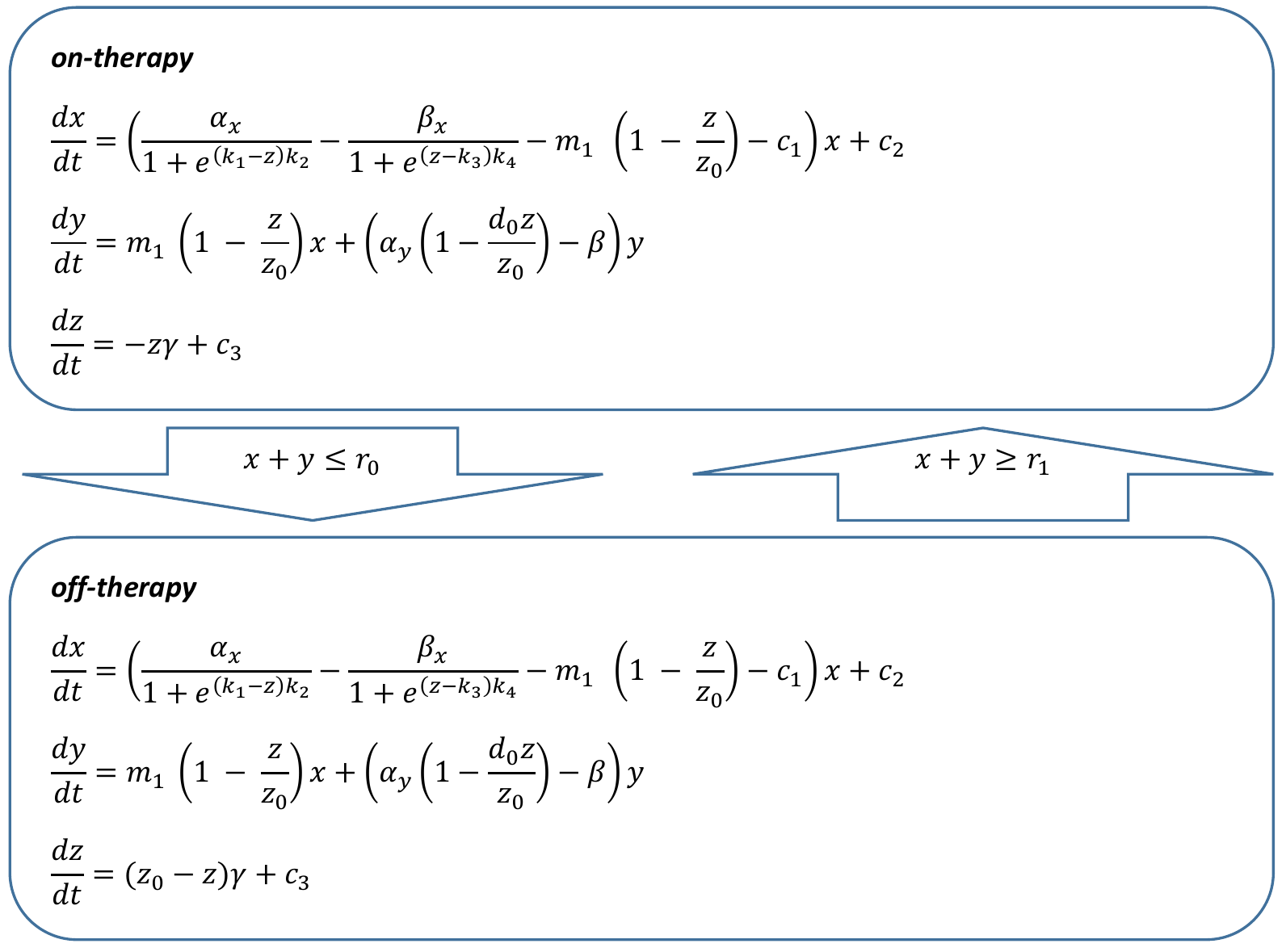}
\caption{Personalized prostate cancer therapy model}
\label{fig:cancer_model}
\end{figure}

\subsection{Prostate cancer therapy: \preach\ file}
\begingroup
\fontsize{7pt}{7pt}\selectfont
\begin{Verbatim}[frame=single, numbers=left, samepage = true]
// This is a pdrh file corresponding to the prostate cancer therapy model with 
// one random and one nondeterministic parameter.
MODEL_TYPE(NPHA) // defining model type
#define betax 0.0175
#define betay 0.0168
#define k1 10.0
#define k2 1.0
#define k3 10.0
#define k4 2
#define m1 0.00005
#define z0 12.0
#define gamma 0.08
#define r1 10.0
#define r0 4.0
#define d0 1.0
#define c1 0.01
#define c2 0.03
#define c3 0.02
#define Gx ((alphax/(1+exp((k1-z)*k2)))-(betax/(1+exp((z-k3)*k4))))
#define Gy ((alphay * (1 - (d0 * (z / z0)))) - betay)
#define Mxy (m1 * (1 - (z / z0)))
#define scale 1.0
#define T 100.0
N(0.05,0.01)alphay; // random parameter, normally distributed
[0,T]time;
[0,T]tau;
[0,100.0]x;
[0,10.0]y;
[0.0,100.0]z;
[0.0197,0.0204]alphax; // nondeterministic parameter
{ 
mode1; // on-therapy
	invt:
		(y <= 1);
	flow:
		d/dt[x]=scale * ((Gx - Mxy - c1) * x + c2);
		d/dt[y]=scale * (Mxy * x + Gy * y);
		d/dt[z]=scale * (-z * gamma + c3);
		d/dt[tau]=scale * 1.0;
	jump:
		((x+y)=r0)==>@2(and(tau'=tau)(x'=x)(y'=y)(z'=z));
}
{ 
mode2; // off-therapy
	invt:
		(y <= 1);
	flow:
		d/dt[x]=scale * ((Gx - Mxy - c1) * x + c2);
		d/dt[y]=scale * (Mxy * x + Gy * y);
		d/dt[z]=scale * ((z0 - z) * gamma + c3);
		d/dt[tau]=scale * 1.0;
	jump:
		((x+y)=r1)==>@1(and(tau'=tau)(x'=x)(y'=y)(z'=z));
}
init:
	@1(and (x = 19) (y = 0.1) (z = 12.5) (tau = 0));
goal: // unsafe region
	@2(and(y <= 1)(tau = T));
goal_c: // unsafe region complement
	@2(and(y > 1.0)(tau < T));
\end{Verbatim}
\endgroup

\section{$\delta$-satisfiability}\label{apndx:delta-sat}

In order to overcome the undecidability of reasoning about general real formulae, 
Gao {\em et al.} recently
defined the concept of $\delta$-satisfiability over the reals \cite{DBLP:conf/lics/GaoAC12}, 
and presented a corresponding $\delta$-complete decision procedure. 
The main idea is to decide correctly whether slightly {\em relaxed} sentences over the reals 
are satisfiable or not. The following definitions are from \cite{DBLP:conf/lics/GaoAC12}.
\begin{definition}
A bounded quantifier is one of the following:
\[ \begin{array}{ll}
	\exists^{[a, b]} x  = \exists x : (a \le x \wedge x \le b) \\[1ex]
	\forall^{[a, b]} x  = \forall x : (a \le x \wedge x \le b)
\end{array} \]
\end{definition}

\begin{definition}
A bounded $\Sigma_1$ sentence is an expression of the form:
\begin{equation*} 
	\exists^{I_{1}} x_{1}, ..., \exists^{I_{1}} x_{n}: \psi(x_{1}, ..., x_{n})
\end{equation*}
where $I_{i} = [a_{i}, b_{i}]$ are intervals, $\psi(x_{1}, ..., x_{n})$ is a Boolean combination
of atomic formulas of the form $g(x_{1}, ..., x_{n})\, \texttt{op} \, 0$, where $g$ is a composition
of Type 2-computable functions and $\texttt{op} \in \{<, \le, >, \ge, =, \ne\}$.
\end{definition}
Any bounded $\Sigma_1$ sentence is equivalent to a $\Sigma_1$ sentence in which all the atoms are
of the form $f(x_{1}, ..., x_{n}) = 0$ (\ie, the only {\em \texttt{op}} needed is `=')
\cite{DBLP:conf/lics/GaoAC12}.
Essentially, Type 2-computable functions can be approximated arbitrarily well by finite
computations of a special kind of Turing machines (Type 2 machines); most of the `useful' functions 
over the reals (\eg, continuous functions) are Type 2-computable \cite{kobook}.

The notion of $\delta$-weakening \cite{DBLP:conf/lics/GaoAC12} of a bounded sentence is central
to $\delta$-satisfiability.
\begin{definition}
Let $\delta \in \mathbb{Q^+} \cup \{0\}$ be a constant and $\phi$ a bounded $\Sigma_1$-sentence 
in the standard form 
\begin{equation} \label{def:sigma1}
\phi = \exists^{I_{1}} x_{1}, ..., \exists^{I_{n}} x_{n}: \bigwedge^{m}_{i = 1}(\bigvee^{k_i}_{j = 1} f_{ij}(x_{1}, ..., x_{n}) = 0)
\end{equation}
where $f_{ij}(x_{1}, ..., x_{n}) = 0$ are atomic formulas. The $\delta$-{\em weakening} of $\phi$ 
is the formula:
\begin{equation*}
\phi^\delta = \exists^{I_{1}} x_{1}, ..., \exists^{I_{n}} x_{n}: \bigwedge^{m}_{i = 1}(\bigvee^{k_i}_{j = 1} |f_{ij}(x_1, ..., x_n)| \le \delta)
\end{equation*}
\end{definition}
Note that $\phi$ implies $\phi^\delta$, while the converse is obviously not true.
The bounded $\delta$-satisfiability problem asks for the following: given a sentence of 
the form (\ref{def:sigma1}) and $\delta \in \mathbb{Q^+}$, correctly decide between
\begin{itemize}
	\item \textbf{unsat}: $\phi$ is false,
	\item \textbf{$\delta$-true}: $\phi^\delta$ is true.
\end{itemize}
If the two cases overlap (\ie, $\phi$ is both false and $\delta$-satisfiable) then either decision 
can be returned, thereby causing a {\em false alarm}. Such a scenario reveals that the 
formula is {\em fragile} --- a small perturbation (\ie, a small $\delta$) can change the formula's
truth value. 
The dReal tool \cite{DBLP:conf/cade/GaoKC13} implements an algorithm for solving 
the $\delta$-satisfiability problem, \ie, a $\delta$-complete decision procedure. Basically, 
the algorithm combines a DPLL procedure (for handling the Boolean parts of the formula) with 
interval constraint propagation (for handling the real arithmetic atoms).